\documentclass[11pt,letterpaper]{article}

\usepackage[margin=1in]{geometry}
\usepackage[T1]{fontenc}
\usepackage{lmodern}
\usepackage{microtype}
\usepackage{xurl}

\usepackage{amsmath,amssymb,amsthm,mathtools}

\usepackage{booktabs}
\usepackage{array}
\usepackage{adjustbox}
\usepackage{enumitem}

\usepackage{float}
\floatstyle{ruled}
\newfloat{algorithm}{tbp}{loa}
\floatname{algorithm}{Algorithm}
\newenvironment{algorithmic}[1][]{\begin{enumerate}[leftmargin=2em,label=\arabic*.,nosep]}{\end{enumerate}}
\newcommand{\State}{\item}
\newcommand{\Statex}{\item[]}
\newcommand{\Return}{\textbf{return} }

\usepackage{xcolor}
\usepackage{tikz}
\usetikzlibrary{arrows.meta,positioning,calc,fit}
\usepackage{subcaption}

\usepackage[numbers,sort&compress]{natbib}
\usepackage[
    colorlinks=true,
    linkcolor=blue,
    citecolor=blue,
    urlcolor=blue
]{hyperref}
\usepackage[capitalize,nameinlink]{cleveref}

\newtheorem{theorem}{Theorem}[section]
\newtheorem{lemma}[theorem]{Lemma}

\theoremstyle{definition}

\theoremstyle{remark}

\crefname{theorem}{Theorem}{Theorems}
\Crefname{theorem}{Theorem}{Theorems}
\crefname{lemma}{Lemma}{Lemmas}
\Crefname{lemma}{Lemma}{Lemmas}
\crefname{proposition}{Proposition}{Propositions}
\Crefname{proposition}{Proposition}{Propositions}
\crefname{corollary}{Corollary}{Corollaries}
\Crefname{corollary}{Corollary}{Corollaries}
\crefname{definition}{Definition}{Definitions}
\Crefname{definition}{Definition}{Definitions}
\crefname{remark}{Remark}{Remarks}
\Crefname{remark}{Remark}{Remarks}
\crefname{algorithm}{Algorithm}{Algorithms}
\Crefname{algorithm}{Algorithm}{Algorithms}

\newcommand{\OPT}{\mathrm{OPT}}
\newcommand{\Scal}{\mathcal{S}}
\allowdisplaybreaks

\usepackage{authblk}

\newcommand{\B}{\mathcal B}
\newcommand{\Ccal}{\mathcal C}
\newcommand{\E}{\mathbb E}
\newcommand{\poly}{\operatorname{poly}}

\title{Coordinating Boundaries: Improved Approximation Algorithms for Unweighted Non-Preemptive Throughput Maximization}

\author[1]{%
Weitian Tong\thanks{\textbf{Corresponding author}. 
Email: \texttt{wtong.research@gmail.com}}%
}

\author[1]{%
Yao Xu
}

\affil[1]{%
\parbox{0.5\textwidth}{%
\centering
School of Computing, Georgia Southern University\\
1332 Southern Dr, Statesboro, Georgia 30458, USA
}%
}

\begin{document}
\raggedbottom
\maketitle

\begin{abstract}
We study unweighted nonpreemptive throughput maximization: 
given jobs with processing times, release times, and deadlines, 
schedule a maximum-cardinality subset on identical machines. 
For every fixed $\varepsilon>0$, we give randomized algorithms with expected approximation ratios
$17/13+\varepsilon$ in polynomial time and $21/17+\varepsilon$ in pseudopolynomial time. 
These improve the respective $4/3+\varepsilon$ and $5/4+\varepsilon$ guarantees of 
Armbruster, Grandoni, Tinguely, and Wiese [STOC'26]. 
Both results hold for an arbitrary number of identical machines.

Our algorithms build on their configuration-LP and matching framework, which partitions time into blocks. 
For a job whose window meets multiple blocks, its boundary blocks are the first and last such blocks. 
Crossing boundary pairs, whose endpoints strictly interleave, 
permit additional assignments to sampled execution slots, 
while a noncrossing family can be coordinated using consistent joint distributions on at most three blocks.
The exchange analysis accounts for fractional job assignments that must be erased before inserting an additional job.
For the pseudopolynomial implementation, we project configurations onto selected job identities 
and use separate configuration distributions with common placement marginals. 
This retains efficient pricing while enforcing the correlations needed for coordination.
\end{abstract}

\clearpage

\section{Introduction}\label{sec:introduction}

\emph{Nonpreemptive throughput maximization} is a fundamental scheduling problem,
with applications in real-time computing and communication resource allocation
\citep{BarNoyGuhaNaorSchieber2001,BarNoyBarYehudaFreundNaorSchieber2001}.
We are given $n$ jobs and $m$ identical machines.
Each job $j$ has \emph{processing time} $p_j>0$ on every machine,
\emph{release time} $r_j\ge0$, and \emph{deadline} $d_j>r_j$.
Its \emph{time window} is $[r_j,d_j)$.
A \emph{feasible schedule} executes each selected job once, continuously for $p_j$
time units within its window, with no overlap on the same machine.
We study the \emph{unweighted} problem: 
maximize the number of completed jobs, called the \emph{throughput}.

All time data are integers encoded in binary. Let $T=\max_j d_j$ be the
\emph{time horizon}, $N$ the input length, and $\OPT$ the maximum throughput.
An \emph{expected $\rho$-approximation} is a randomized algorithm that 
always returns a feasible schedule and whose expected throughput is at least $\OPT/\rho$.
We measure \emph{polynomial time} in $N$;
\emph{pseudopolynomial time} also permits polynomial dependence on $T$.
Throughout, $\varepsilon>0$ is an arbitrarily small fixed \emph{accuracy parameter}.
Our main result is the following.

\begin{theorem}\label{thm:main}
For every fixed $\varepsilon>0$, unweighted nonpreemptive throughput maximization
on any number of identical machines admits randomized algorithms with 
expected approximation ratios $17/13+\varepsilon$ in polynomial time 
and $21/17+\varepsilon$ in pseudopolynomial time. 
Every returned schedule is feasible.
The respective running times are $n^{O_\varepsilon(1)}\poly(N)$ and $(nT)^{O_\varepsilon(1)}\poly(N)$.
\end{theorem}

We describe the main historical ideas first for one machine;
the subsequent guarantees also extend to any number of identical machines.

Spieksma's greedy algorithm \citep{Spieksma1999} gave the first $2$-approximation.
This greedy routine underlies the block construction described next.
%
\citep{ChuzhoyOstrovskyRabani2006} improved the approximation ratio to
$e/(e-1)+\varepsilon\approx1.582+\varepsilon$ 
(in a more general setting where each job's possible execution intervals are listed explicitly).
Their greedy-based reduction partitions time into disjoint intervals, 
called \emph{blocks}, and fixes schedules in some blocks. 
These choices preserve a \emph{nearly optimal schedule}, losing an arbitrarily small fraction of $\OPT$,
that crosses no block boundary and schedules $O_\varepsilon(1)$ jobs per unresolved block.
A \emph{configuration} is a feasible job set of this bounded size inside a block;
we retain a schedule witnessing its feasibility.
For the residual instance, let $\Ccal(B)$ be the configurations of block $B$.
The configuration linear program (LP) has nonnegative variables $x_{B,C}$.
Assigning unused mass to the empty configuration gives $\sum_{C\in\Ccal(B)}x_{B,C}=1$, 
making $x_{B,C}$ the probability of choosing $C$ in $B$.
A job's \emph{placement marginal} and total \emph{fractional mass} are
$y_{j,B}:=\sum_{\substack{C\in\Ccal(B)\\j\in C}}x_{B,C}$, $y_j:=\sum_B y_{j,B}\le1.$
The LP maximizes $\sum_j y_j$.
Its rounding procedure uses \emph{independent configuration sampling}: 
choose one configuration per block according to $x_{B,C}$, independently across blocks, 
and keep one copy of each job.
Jobs within a configuration are selected together. Thus $y_j$ counts expected copies,
whereas the probability that $j$ contributes is
$\Pr[j\text{ is retained}]=1-\prod_B(1-y_{j,B})\ge(1-1/e)y_j.$
Adding the fixed schedules yields the claimed factor.
The loss comes from duplicate selections:
if unit mass is spread equally over $h$ blocks, 
the retained probability is $1-(1-1/h)^h$, approaching $1-1/e$ as $h$ grows.
This yields the mentioned approximation ratio of $e/(e-1)+\varepsilon$.

\citep{ImLiMoseley2020} improved the approximation ratio to $1/0.6448+\varepsilon<1.551+\varepsilon$.
Their LP represents alternative block intervals, and its solution determines a partition into \emph{base windows}. 
For each job, its \emph{interior base windows} are those strictly 
between the first and last base windows intersecting its time window; 
every such interior window is contained in the job's window. 
Their rounding combines independent configuration sampling with a procedure that 
independently assigns each job to at most one interior base window and then packs the assigned jobs. 
After excluding long placements and reserving capacity, 
this procedure recovers the fractional mass assigned to interior base windows 
up to an arbitrarily small additive loss relative to the LP value. 
Returning the better schedule improves the retained fraction beyond $1-1/e$.

\citep{ArmbrusterGrandoniTinguelyWiese2026} further improved the approximation ratios to
$4/3+\varepsilon$ in polynomial time and $5/4+\varepsilon$ in pseudopolynomial time.
They build on \citep{ChuzhoyOstrovskyRabani2006}'s block construction, configuration LP, and independent sampling,
and on \citep{ImLiMoseley2020}'s complementary procedure that assigns jobs to blocks individually.
Their new ingredients improve both rounding procedures.
They reinterpret each sampled execution as a \emph{slot}, an interval on a specified machine.
A job is \emph{compatible} with a slot if it fits inside both the slot and its window.
A maximum \emph{job--slot matching} selects as many compatible pairs as possible,
using each job and each slot at most once. 
A duplicate occurs when the same job is selected in multiple sampled configurations. 
Its spare slot can therefore accommodate another job that was not selected (Figure~\ref{fig:intro-agtw}).

To analyze this reuse, 
\citep{ArmbrusterGrandoniTinguelyWiese2026} strengthen the decomposition.
A job's \emph{boundary blocks} are the first and last blocks intersecting its window.
A \emph{superblock} is a union of consecutive blocks, 
and a job \emph{spans} it when its window contains the entire superblock.
A nearly optimal schedule executes each job within a boundary block or within
one block of a spanned superblock.
Within each superblock, \emph{fractional grouping} orders jobs by processing time
and partitions their mass into nearly equal groups: 
shorter jobs can use slots of longer jobs in a superblock they span. 
With many blocks per superblock, concentration bounds 
control the expected loss from insufficient slots.

\begin{figure}[htbp]
\centering
\begin{tikzpicture}[x=0.78cm,y=0.82cm,font=\small,>=Stealth]
  \fill[green!8] (4,-0.1) rectangle (8,2.8);
  \foreach \i in {1,...,6}{
    \pgfmathsetmacro{\a}{2*\i-2}
    \draw[gray!55] (\a,2.05) rectangle +(2,0.5);
    \node at (\a+1,2.3) {$B_{\i}$};
  }
  \draw[green!45!black,thick] (4,2.72) -- (4,2.85) -- (8,2.85) -- (8,2.72);
  \node[above,green!35!black] at (6,2.9) {spanned superblock $S=[4,8)$};
  \draw[orange!75!black,thick,|-|] (1,3.65) -- (11,3.65);
  \node[above] at (6,3.65) {window of $j$: $[1,11)$};

  \node[anchor=east] at (-0.15,1.45) {Sampled executions};
  \draw[gray!40] (0,1.18) -- (12,1.18);
  \draw[fill=blue!22,draw=blue!65!black] (0,1.18) rectangle (2,1.72);
  \draw[fill=blue!22,draw=blue!65!black] (4,1.18) rectangle (6,1.72);
  \node at (1,1.45) {$k$};
  \node at (5,1.45) {$k$};

  \node[anchor=east] at (-0.15,0.45) {After matching};
  \draw[gray!40] (0,0.18) -- (12,0.18);
  \draw[fill=blue!22,draw=blue!65!black] (0,0.18) rectangle (2,0.72);
  \node at (1,0.45) {$k$};
  \draw[gray!70,dashed] (4,0.18) rectangle (6,0.72);
  \draw[fill=orange!30,draw=orange!75!black] (4,0.18) rectangle (5,0.72);
  \node at (4.5,0.45) {$j$};
  \draw[->,gray!75] (0,-0.35) -- (12.3,-0.35) node[right] {time};
  \foreach \t in {0,2,4,6,8,10,12}{
    \draw[gray!75] (\t,-0.4) -- (\t,-0.3);
    \node[below,font=\footnotesize] at (\t,-0.4) {$\t$};
  }
\end{tikzpicture}
\caption{Reusing a duplicate's slot recovers an unsampled job.
Here $p_k=2$ with window $[0,12)$, and $p_j=1$ with window $[1,11)$.
Matching keeps $k$ at $[0,2)$ and places $j$ at $[4,5)$ inside the spare slot $[4,6)$.
Because $j$ spans $S$, it fits any sufficiently long slot within $S$.}
\label{fig:intro-agtw}
\end{figure}
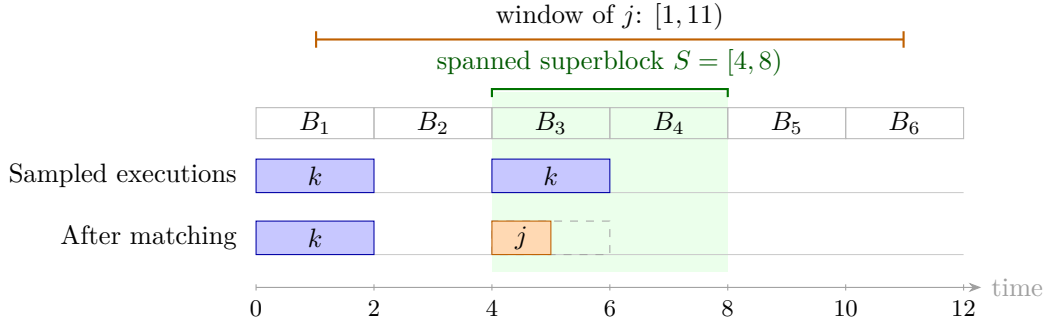

A job is \emph{local} if its window lies in one block, and \emph{global} otherwise.
Let $L$ and $G$ be their respective fractional masses, with total LP value $F=L+G$.
Independent sampling retains local mass and 
at least $3/4$ of the mass of jobs supported only at their boundaries. 
Slot reuse extends the latter guarantee to all global jobs. 
Thus matching retains expected throughput $L+3G/4$, 
up to an arbitrarily small loss relative to $F$, giving the $4/3+\varepsilon$ factor.

To obtain $5/4+\varepsilon$, 
\citep{ArmbrusterGrandoniTinguelyWiese2026} 
extended \citep{ImLiMoseley2020}'s individual assignment rule to include boundary blocks. 
They merged more blocks, increasing the average number of jobs per block in a nearly optimal schedule. 
Each global job chooses at most one block. 
Jobs long relative to the length of $B\cap[r_j,d_j)$ and assignments
violating load constraints are discarded, at a small expected loss relative to $F$.
The output recovers almost all global mass. 
Up to these small losses, the larger complete schedule therefore has expected throughput at least
$\max\left\{L+\frac34G,\,G\right\}\ge\frac45F$.
This construction needs configurations of $O_\varepsilon(\log(T+1))$ jobs.
\citep{ArmbrusterGrandoniTinguelyWiese2026}  solved the resulting LP in pseudopolynomial time 
using \emph{color coding}, which seeks jobs with distinct temporary colors \citep{AlonYusterZwick1995}, 
and dynamic programming over partial schedules.
\citep{ImLiMoseley2020}'s near-optimal algorithm for sufficiently many machines allows 
the later fixed-machine guarantees to extend to arbitrary $m$.

A remaining loss arises from independent boundary sampling: 
a job with marginal $1/2$ at each boundary is selected with probability only $3/4$.
This bound does not use the relative order of different jobs' boundary blocks.
Our approach exploits the relative order of boundary blocks.
As illustrated in Figure~\ref{fig:approach-crossing}, 
when two jobs have interleaving boundaries, 
a shorter job missed at both boundaries can be assigned to a spare slot of a longer job selected at both.
We further coordinate configuration sampling across blocks to reduce duplicate selections, 
improving both approximation guarantees of \citep{ArmbrusterGrandoniTinguelyWiese2026}.

\subsection{Sketch of our approach}

We retain the blocks, superblocks, and job--slot matching of \citep{ArmbrusterGrandoniTinguelyWiese2026} 
and exploit the order of boundary blocks in two complementary ways. 
The first output samples configurations independently and 
uses maximum job--slot matching to reassign spare slots 
from duplicate selections to other compatible jobs, 
with the boundary order enabling additional reassignments.
The second jointly samples configurations to prevent duplicate boundary selections for certain jobs 
while preserving each block's configuration distribution.
Thus, the first strategy exploits duplicates, 
whereas the second avoids them; 
we return the larger of the two schedules.
It suffices to consider fixed $m$; 
Section~\ref{sec:completion} gives the reduction using the large-machine algorithm of \citep{ImLiMoseley2020}.

Both outputs use a stronger configuration LP that 
also specifies how configurations on up to three blocks can be sampled jointly, 
with agreement on overlapping groups.
Fix a solution of this relaxation with value $F=L+G$. 
For each global job $j$, let $a_j,b_j$ be its left and right boundary marginals, 
and let $z_j$ be its mass in all other eligible blocks. Thus $a_j+b_j+z_j=y_j\le1$.
All job sums in this overview range over global jobs. 
Define $W:=\sum_j a_jb_j$ and $\mathcal D:=\frac G4-W\ge0$.
Under independent sampling, $a_jb_j$ is the probability that job $j$ is selected at both boundary blocks, 
resulting in a duplicate selection.
Since $a_jb_j \le (a_j+b_j)/4 \le y_j/4$, 
the total expected number of such collisions satisfies $W \le G/4$.
Thus, $\mathcal D = G/4-W$ measures the slack in this worst-case bound,
which arises from interior LP mass and unequal boundary marginals.
Let $\eta>0$ be an internal accuracy parameter, 
chosen sufficiently small in terms of $\varepsilon$. 
The partition and rounding losses below are $O(\eta F)$.

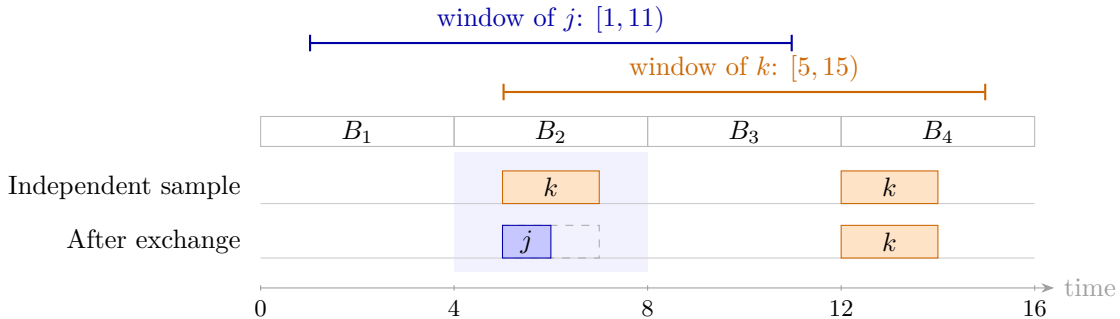
\begin{figure}[htbp]
\centering
\begin{tikzpicture}[x=0.64cm,y=0.72cm,font=\small,>=Stealth]
  \draw[blue!65!black,thick,|-|] (1,4.25) -- (11,4.25);
  \node[above,blue!65!black] at (6,4.25) {window of $j$: $[1,11)$};
  \draw[orange!80!black,thick,|-|] (5,3.35) -- (15,3.35);
  \node[above,orange!80!black] at (10,3.35) {window of $k$: $[5,15)$};
  \foreach \i in {1,...,4}{
    \pgfmathsetmacro{\a}{4*\i-4}
    \draw[gray!55] (\a,2.35) rectangle +(4,0.55);
    \node at (\a+2,2.625) {$B_{\i}$};
  }
  \fill[blue!5] (4,0.05) rectangle (8,2.25);
  \node[anchor=east] at (-0.2,1.6) {Independent sample};
  \draw[gray!40] (0,1.3) -- (16,1.3);
  \foreach \a in {5,12}{
    \draw[fill=orange!22,draw=orange!80!black]
      (\a,1.3) rectangle +(2,0.6);
    \node at (\a+1,1.6) {$k$};
  }
  \node[anchor=east] at (-0.2,0.6) {After exchange};
  \draw[gray!40] (0,0.3) -- (16,0.3);
  \draw[gray!65,dashed] (5,0.3) rectangle (7,0.9);
  \draw[fill=blue!22,draw=blue!65!black] (5,0.3) rectangle (6,0.9);
  \node at (5.5,0.6) {$j$};
  \draw[fill=orange!22,draw=orange!80!black] (12,0.3) rectangle (14,0.9);
  \node at (13,0.6) {$k$};
  \draw[->,gray!75] (0,-0.25) -- (16.4,-0.25) node[right] {time};
  \foreach \t in {0,4,8,12,16}{
    \draw[gray!75] (\t,-0.3) -- (\t,-0.2);
    \node[below,font=\footnotesize] at (\t,-0.3) {$\t$};
  }
\end{tikzpicture}
\caption{A crossing permits reuse of a boundary slot.
Here $p_j=1$, $p_k=2$, and the boundary pairs are $(B_1,B_3)$ and $(B_2,B_4)$.
The sample selects $k$ twice and selects $j$ at neither boundary; other jobs are omitted.
Keep $k$ at $[12,14)$ and place $j$ at $[5,6)$ inside its spare slot $[5,7)$.
Any fractional interior assignment of $j$ must first be erased from the analysis.}
\label{fig:approach-crossing}
\end{figure}

Index the blocks from left to right and regard them as vertices of a graph.
Each global job contributes an edge joining its two boundary blocks; 
we call this its \emph{boundary edge}. 
Jobs with the same boundary blocks contribute separate edges.
Two boundary edges \emph{cross} when their endpoints strictly interleave:
jobs $j,k$ have boundary pairs $(B_u,B_v)$ and $(B_s,B_t)$ with $u<s<v<t$.
Suppose $p_j\le p_k$. The entire block $B_s$ lies inside $j$'s window, 
so $j$ fits any slot of $k$ there. 
If $j$ is absent at both boundaries and $k$ appears at both,
we can retain $k$ in $B_t$ and place $j$ in $k$'s spare slot in $B_s$.
The case $p_k<p_j$ is symmetric; see Figure~\ref{fig:approach-crossing}.

The exchange must account for the \emph{fractional matching} used to recover interior mass: 
fractional job--slot assignments totaling at most one at each job and slot.
This matching uses no slot originating from a job at its own boundary.
Although $j$ is absent at its boundaries, 
this matching may already assign it to interior slots. 
We erase all those assignments before inserting $j$ integrally.
A pointwise bound on the erased fractional mass guarantees an expected
net gain of at least $(a_jb_j)(a_kb_k)$.
Exchanges involving disjoint jobs can be performed simultaneously,
even if their pairs share blocks.
These exchanges provide a lower bound on the maximum matching's value;
neither the fractional assignments nor the exchanges need to be constructed by the algorithm.

To connect the exchange analysis of the first output with the joint sampling used for the second,
call a job \emph{good} when $a_j,b_j\ge1/4$. 
Form the \emph{crossing graph}, whose vertices are good jobs and whose edges join crossing pairs. 
Choose a maximal matching, 
and let $\nu$ be its size.
The unmatched good jobs have \emph{noncrossing} boundary edges: 
no two strictly interleave, although nesting and shared endpoints are allowed.
Otherwise their crossing pair could be added to the matching.

To avoid selecting any unmatched good job at both boundary blocks,
we strengthen the configuration LP
by introducing a \emph{joint distribution} over tuples of configurations
for every set of at most three blocks, with one configuration per block.
Each tuple assigned positive probability must consist of pairwise disjoint job sets,
so no job is selected in two blocks of the same tuple.
The distribution induced on each individual block must equal its configuration distribution in the LP. 
We also require these joint distributions to be \emph{consistent}:
whenever two block sets overlap, they induce the same joint distribution on the shared blocks.
The original LP does not guarantee that 
its block distributions extend to joint distributions satisfying all these requirements.
The nearly optimal schedule supplied by the partition, however, 
gives a feasible solution to the stronger LP: 
for each block set, assign probability one to the tuple of configurations used by that schedule. 
These choices satisfy both disjointness and consistency. 
Thus the strengthened LP has value at least that of this nearly optimal schedule.

\begin{figure}[htbp]
\centering
\begin{tikzpicture}[x=1cm,y=1cm,font=\small,>=Stealth,
  triple/.style={draw,rounded corners=2pt,minimum width=2.2cm,
                 minimum height=0.58cm,inner sep=4pt}]
  \coordinate (b1) at (0,1.8);
  \coordinate (b2) at (1.5,2.8);
  \coordinate (b3) at (3,1.8);
  \coordinate (b4) at (3,0);
  \coordinate (b5) at (1.5,-1);
  \coordinate (b6) at (0,0);
  \fill[gray!7] (b1) -- (b2) -- (b3) -- cycle;
  \fill[green!8] (b3) -- (b4) -- (b5) -- cycle;
  \fill[gray!7] (b1) -- (b5) -- (b6) -- cycle;
  \fill[blue!9] (b1) -- (b3) -- (b5) -- cycle;
  \draw[gray!65] (b1) -- (b2) -- (b3) -- (b4) -- (b5) -- (b6) -- cycle;
  \draw[blue!65!black,thick] (b1) -- (b3) -- (b5) -- cycle;
  \foreach \i in {1,...,6}{\fill (b\i) circle (1.6pt);}
  \node[above left] at (b1) {$B_1$};
  \node[above] at (b2) {$B_2$};
  \node[above right] at (b3) {$B_3$};
  \node[right] at (b4) {$B_4$};
  \node[below] at (b5) {$B_5$};
  \node[left] at (b6) {$B_6$};
  \node[above] at (1.5,3.15) {Noncrossing boundary edges};

  \node[triple,fill=blue!9] (root) at (7.7,2.25) {$\{B_1,B_3,B_5\}$};
  \node[triple,fill=gray!7] (left) at (5.1,-0.1) {$\{B_1,B_2,B_3\}$};
  \node[triple,fill=green!8] (middle) at (7.7,-0.1) {$\{B_3,B_4,B_5\}$};
  \node[triple,fill=gray!7] (right) at (10.3,-0.1) {$\{B_1,B_5,B_6\}$};
  \draw[->] (root) -- node[pos=0.62,fill=white,inner sep=2pt]
    {$\{B_1,B_3\}$} (left);
  \draw[->] (root) -- node[pos=0.70,fill=white,inner sep=2pt]
    {$\{B_3,B_5\}$} (middle);
  \draw[->] (root) -- node[pos=0.62,fill=white,inner sep=2pt]
    {$\{B_1,B_5\}$} (right);
  \node[above] at (7.7,3.15) {Condition on each shared pair};
  \node[align=center] at (7.7,-0.85)
    {Each node is a triangle; edge labels give its shared blocks.};
\end{tikzpicture}
\caption{Consistent three-block distributions coordinate a noncrossing family.
The blue chords are retained boundary edges; the blocks are placed cyclically in
their time order. The tree on the right records which triangles share an edge.
After sampling the central triangle, sample each neighboring triangle conditionally
on its shared pair. The same construction applies to a larger triangulation.}
\label{fig:approach-gluing}
\end{figure}
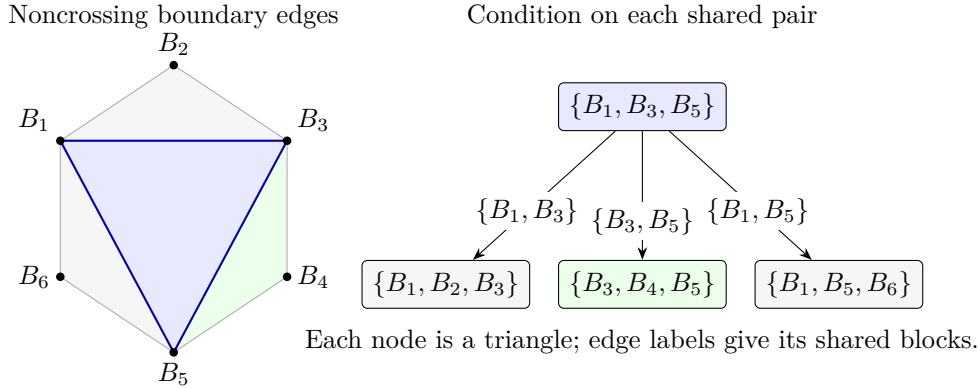

Place the block vertices around a polygon in time order. 
The noncrossing boundary edges extend to a \emph{triangulation}, 
a subdivision into triangles by nonintersecting diagonals. 
Each triangle has a three-block distribution. 
Sample one triangle and extend across adjacent triangles, 
conditioning on the configurations already chosen at their two shared blocks (Figure~\ref{fig:approach-gluing}). 
Agreement on the complete joint distribution of these two blocks makes each extension valid.
This \emph{coordinated sampling} preserves every block's configuration distribution
and prevents duplicate selection at the boundaries of every unmatched good job.

To obtain the second output from these jointly sampled configurations,
retain all sampled local executions and discard global interior executions.
Each unmatched good job may be retained at either boundary, with total probability $a_j+b_j$. 
Restrict each other global job to its more likely boundary. 
Restricting matched jobs loses at most one unit of expected throughput per pair. 
Deleting interior executions and restricting the other global jobs loses at most $4\mathcal D$.

When both jobs have marginal $1/2$ at each boundary, 
the product bound gives an expected gain of $1/16$.
The slack $\mathcal D$ absorbs the decrease in these product gains for other marginals.
Let $A_\times$ and $A_{\rm keep}$ denote the cardinalities of the independent and
coordinated outputs. The analysis in Section~\ref{sec:rounding} gives
\[
  \begin{aligned}
    \E \left[ A_\times \right]
      &\ge L+\frac34G+\frac{\mathcal D}{4}+\frac{\nu}{16}-O(\eta F),\\
    \E \left[ A_{\rm keep} \right] 
      &\ge L+\frac34G+\frac G4-4\mathcal D-\nu-O(\eta F).
  \end{aligned}
\]
The first bound gains from the chosen crossings; 
the second pays for restricting the same jobs to one boundary. 
The weights $16/17$ and $1/17$ cancel both $\mathcal D$ and $\nu$.
Since returning the larger schedule dominates this convex combination,
\[
  \E \left[ \max\{A_\times,A_{\rm keep}\} \right]
  \ge \frac{16}{17}\E \left[ A_\times \right] +\frac1{17}\E \left[ A_{\rm keep} \right]
  \ge L+\frac{13}{17}G-O(\eta F).
\]
The near-optimal fractional value then gives the $17/13+\varepsilon$ factor.
In pseudopolynomial time, the complementary global-job output of
\citep{ArmbrusterGrandoniTinguelyWiese2026} also retains expected throughput $G-O(\eta F)$. 
Taking the largest of the three schedules therefore guarantees
$
  \max\left\{L+\frac{13}{17}G,\,G\right\}-O(\eta F)
  \ge \frac{17}{21}F-O(\eta F),
$
giving the $21/17+\varepsilon$ factor.

It remains to implement the stronger relaxation. 
Let $K$ bound the number of jobs in a configuration. 
For $K=O_{m,\eta}(1)$, we enumerate all configurations and tuples on at most three blocks, 
obtaining a polynomial-size LP.
The pseudopolynomial algorithm needs $K=O_{m,\eta}(\log(T+1))$,
for which explicit enumeration no longer gives the desired running time.
The main challenge is to enforce the correlations between configurations of different blocks required 
by the stronger LP without explicitly enumerating those configurations.
We address this challenge through \emph{projected consistency}:
each three-block law records only which tracked job identities its configurations contain.
At an update, we track boundary identities with $y_{j,B}\ge1/8$, at most $8K$ per block.
This leaves only $2^{O(K)}$ projected states per group of at most three blocks.
Each update adds a separate configuration copy linked only to its new projected constraints; 
earlier copies and constraints remain. 
All copies share the marginals $y_{j,B}$, while retaining their own configuration distributions.
We coordinate the projected states first, then sample full configurations conditionally on those states.

Consequently, \emph{pricing}, which finds a configuration of maximum prescribed score,
has additive job weights, as in \citep{ArmbrusterGrandoniTinguelyWiese2026}, plus
one additional term depending on the selected tracked subset.
We enumerate that subset, exclude the other tracked jobs, 
and extend color coding to require every prescribed identity, including jobs of negative weight.
Reoptimization can change which marginals are large. 
A small quadratic penalty on the placement marginals makes 
substantial changes between successive optima decrease the objective, 
bounding the number of such updates by a polynomial.
We implement this penalty through a rational piecewise-linear approximation 
and stop when consecutive marginal vectors are sufficiently close.
At termination, the gap between the tracking threshold $1/8$ 
and the good-job threshold $1/4$ ensures that 
at most $\eta F$ current good jobs are untracked at a boundary.
Restricting these jobs to one boundary costs only $O(\eta F)$, 
completing the pseudopolynomial implementation.

Table~\ref{tab:framework} separates the inherited framework from
the changes in the rounding and its implementation.
A visual roadmap of the paper is provided in  Appendix~\ref{app:roadmap}, 
which shows how the output schedules are sampled and constructed,
and how the main lemmas lead to the approximation guarantees.

\begin{table}[htbp]
\centering\small
\renewcommand{\arraystretch}{1.12}
\begin{adjustbox}{max width=\linewidth}
\begin{tabular}{@{}>{\raggedright\arraybackslash}p{.18\linewidth}
                  >{\raggedright\arraybackslash}p{.34\linewidth}
                  >{\raggedright\arraybackslash}p{.42\linewidth}@{}}
\toprule
Ingredient & \citep{ArmbrusterGrandoniTinguelyWiese2026} & This paper\\ 
\midrule
Structural partition &
Blocks and superblocks with bounded-size configurations. &
Same framework; higher average job counts per superblock control interior matching losses.\\ \hline
Independent rounding &
Independent configurations followed by maximum job--slot matching. &
Same matching; stronger analysis via simultaneous crossing exchanges.\\ \hline
Boundary coordination &
Independent block samples in the matching output. &
Additional output coordinates noncrossing boundary edges through consistent
joint distributions.\\ \hline
Pseudopolynomial implementation &
Color coding prices additive job weights without enumerating configurations. &
Projected consistency enforces correlations; color coding with required jobs handles
one additional subset term per configuration copy.\\ \hline
Complementary output &
Global-job assignment and load tests; combined guarantee $4F/5$. &
Same output; the largest of the three schedules gives $17F/21$.\\
\bottomrule
\end{tabular}
\end{adjustbox}
\caption{Comparison with \citep{ArmbrusterGrandoniTinguelyWiese2026}.
The displayed throughput guarantees are in expectation and omit $O(\eta F)$ losses.}
\label{tab:framework}
\end{table}

\subsection{Related work}

%
Related models also capture project selection and scheduling for space missions \citep{HallMagazine1994}.
Motivated by fluctuating cloud resources, 
\citep{MurhekarPurohitSvitkinaVeeWang2025} studied nonpreemptive throughput under time-varying capacity, 
where the number of jobs that can run simultaneously changes over time.

The unweighted problem is strongly NP-hard already on one machine \citep{KaulMnichMolter2024}.
Whether it admits a \emph{polynomial-time approximation scheme} (PTAS), 
namely a polynomial-time $(1+\varepsilon)$-approximation for every fixed $\varepsilon>0$,
remains open; \emph{APX-hardness}, which would exclude such a scheme unless $\mathrm{P}=\mathrm{NP}$, 
has not been established either \citep{ArmbrusterGrandoniTinguelyWiese2026}.
Stronger guarantees are known for restricted processing times.
On one machine, equal-length jobs can be handled exactly in $O(n^5)$ time
\citep{ChrobakDurrJaworKowalikKurowski2006}.
%
\citep{HyattDenesikRahgoshaySalavatipour2024} gave a randomized PTAS 
when both $m$ and the number of distinct processing times are fixed.
For an arbitrary number of machines, however, 
the equal-length case is NP-hard \citep{HeegerMolter2025}.

For \emph{weighted throughput}, each job has a nonnegative profit, 
and the objective is to maximize the total profit of jobs completed by their deadlines.
For $m$ identical machines, 
\citep{BarNoyGuhaNaorSchieber2001} achieved approximation ratio $\frac{(1+1/m)^m}{(1+1/m)^m-1}$,
in polynomial time for unweighted jobs and pseudopolynomial time for weighted jobs.
This ratio tends to $e/(e-1)$ as $m\to\infty$.
%
\citep{BermanDasGupta2000} obtained this ratio 
up to an additive $\varepsilon$ in polynomial time for weighted jobs.
The \emph{local-ratio method}, which decomposes the objective into simpler weight functions, 
gives another polynomial-time approach, with an approximation ratio of $2+\varepsilon$ \citep{BarNoyBarYehudaFreundNaorSchieber2001}.
For weighted jobs with $T$ polynomially bounded in $n$,
\citep{ImLiMoseley2020} gave a randomized polynomial-time $(1+\varepsilon)$-approximation in expectation 
when $m$ is sufficiently large as a function of $\varepsilon$. 
Their unweighted large-machine guarantee permits arbitrary binary time data and is used in Section~\ref{sec:completion}.

For the weighted single-machine problem, 
\citep{KarakostasKolliopoulos2025} gave a PTAS 
when either the number of distinct release times or the number of distinct deadlines is fixed.
Exact and heuristic methods have also been studied, including enumeration,
mixed-integer formulations, and branch-and-bound \citep{HungLinLaiChien2026}.

With \emph{preemption}, a job may be interrupted and resumed later.
On one machine, the unweighted problem is solvable exactly in $O(n^4)$ time \citep{Baptiste1999Preemptive}.
For the weighted case, Lawler's exact dynamic program has pseudopolynomial
dependence on the total integer profit~\citep{Lawler1990}.
Applying the approximation framework of \citep{PruhsWoeginger2007} yields a 
\emph{fully polynomial-time approximation scheme} (FPTAS).
Strong guarantees are also possible without preemption under \emph{speed augmentation}.
For any fixed $0<\varepsilon<1$, 
\citep{ImLiMoseleyTorng2015} gave a quasipolynomial-time algorithm that, 
on $m$ identical machines of speed $1+\varepsilon$, achieves at least a $(1-\varepsilon)$
fraction of the optimum weighted throughput on $m$ unit-speed machines.

\paragraph{Organization.}
Section~\ref{sec:prelim} presents the structural partition 
and the configuration LP with consistency constraints. 
Section~\ref{sec:rounding} analyzes the two rounding procedures and combines their guarantees. 
Section~\ref{sec:implementation} gives the polynomial-time and pseudopolynomial-time implementations 
and proves Theorem~\ref{thm:main}. 
Appendix~\ref{app:partition} establishes the refined density bounds for the partition
and Appendix~\ref{app:complementary} provides the complementary global-job rounding guarantee.
Appendix~\ref{app:roadmap} gives a visual roadmap of the paper,
showing the two output constructions and the main proof dependencies.

\section{Partitions and configuration distributions}\label{sec:prelim}

We discard jobs with $p_j>d_j-r_j$ and handle the empty instance directly.
For any feasible schedule, integral release times and processing times allow
each machine's job order to be shifted to its earliest feasible start times.
The resulting schedule has integer start times and no later completion times.
Fix an internal loss parameter $\eta=1/q\le1/100$, where $q\in\mathbb N$.
Section~\ref{sec:completion} specifies its value in terms of $\varepsilon$.

\subsection{Blocks, superblocks, and alignment}

A \emph{block} is a nonempty half-open interval with integer endpoints.
Let $\B$ be a partition of $[0,T)$ into blocks, indexed as $B_1,\ldots,B_{|\B|}$ from left to right.
A \emph{superblock} is a union of consecutive blocks; 
the superblocks form another partition $\Scal$ of $[0,T)$.
A schedule is \emph{aligned} with $\B$ if each execution lies entirely in one block.

For a job $j$, its \emph{release block} $B_L(j)$ contains $r_j$.
Its \emph{deadline block} $B_R(j)=[s,t)$ is determined by $s<d_j\le t$.
Thus a release time at a partition boundary belongs to the block on its right,
whereas a deadline at a boundary belongs to the block on its left.
A job is \emph{local} if its window is contained in one block, and \emph{global}
otherwise. Equivalently, it is local exactly when $B_L(j)=B_R(j)$.
For a global job, these two distinct blocks are its \emph{boundary blocks}.
A job \emph{spans} an interval if its window contains that interval.

We impose the following \emph{eligibility} restriction in the structural
construction and the relaxation.
A local job is eligible only in the block containing its window.
A global job is eligible in its two boundary blocks and 
in every block belonging to a superblock that it spans.
An \emph{interior placement} of a global job is a placement outside its boundary blocks; 
eligibility requires its entire superblock to lie in the job's window.
Within such a superblock, every interval of length $p_j$ satisfies $j$'s release and deadline constraints.
Boundary blocks remain eligible even when only partly covered by the window,
although an execution must still fit in the usable part of the block.
Figure~\ref{fig:eligibility} illustrates these restrictions.
Eligibility is an additional restriction for the structural construction and
the relaxation; the final matching may use other placements that are feasible in the original instance.

\begin{figure}[htbp]
\centering
\begin{tikzpicture}[x=1.22cm,y=1cm,font=\small]
  \foreach \i in {0,...,8}{
    \fill[gray!12] (\i,0) rectangle +(1,.58);
  }
  \foreach \i in {0,8}{
    \fill[blue!22] (\i,0) rectangle +(1,.58);
  }
  \foreach \i in {3,4,5}{
    \fill[green!20] (\i,0) rectangle +(1,.58);
  }
  \foreach \i in {0,...,8}{
    \draw[gray!65] (\i,0) rectangle +(1,.58);
    \pgfmathtruncatemacro{\j}{\i+1}
    \node at (\i+.5,.29) {$B_{\j}$};
  }
  \foreach \a/\b/\lab in {0/3/S_1,3/6/S_2,6/9/S_3}{
    \draw[thick] (\a,.76)--(\a,.88)--(\b,.88)--(\b,.76);
    \node[above] at ({(\a+\b)/2},.9) {$\lab$};
  }
  \draw[blue!70!black,thick,|-|] (.45,1.65)--(8.55,1.65);
  \node[above] at (4.5,1.65) {window of a global job $j$};
  \node[left] at (.45,1.65) {$r_j$};
  \node[right] at (8.55,1.65) {$d_j$};
  \draw[blue!55,dashed] (.45,1.5)--(.45,.6);
  \draw[blue!55,dashed] (8.55,1.5)--(8.55,.6);
  \node[below,blue!70!black] at (.5,-.05) {$B_L(j)$};
  \node[below,blue!70!black] at (8.5,-.05) {$B_R(j)$};
  \draw[orange!80!black,thick,|-|] (4.15,-.35)--(4.85,-.35);
  \node[below,orange!80!black] at (4.5,-.35) {a local window};
\end{tikzpicture}
\caption{Blocks permitted by the eligibility restriction. The global job $j$
may use its boundary blocks $B_1,B_9$ and the blocks $B_4,B_5,B_6$ in the
spanned superblock $S_2$, subject to scheduling feasibility.
Although its window contains $B_2,B_3,B_7,B_8$, it does not span their
superblocks, so those blocks are ineligible. The local job uses only $B_5$.}
\label{fig:eligibility}
\end{figure}
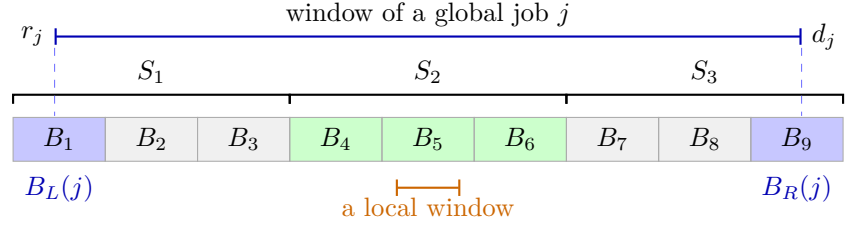

The next lemma specializes the structural framework of 
\citep[Lemmas 4, 6, 27, and 28]{ArmbrusterGrandoniTinguelyWiese2026}.
The construction computes candidate partitions and guarantees a nearly optimal \emph{reference schedule}, 
used only in the analysis, for one of them.
Appendix~\ref{app:partition} proves the required strengthening of the \emph{density bounds}, 
which lower-bound the average number of reference executions per block and per superblock.
The integer parameter $\Delta\ge1$ controls the density requirement and the number of jobs permitted in a block.

\begin{lemma}\label{lem:partition}
For fixed $m,\eta$, there is a computable positive integer $K_0$ depending only on $m,\eta$. 
Define $\Gamma=\frac{2K_0}{\eta^7}$ and $H=\left\lceil5K_0\Delta\Gamma^q\right\rceil$.
If $\OPT\ge H$, a construction with running time polynomial in the input length and $\Delta$, 
for fixed $m,\eta$, produces at most $q$ candidate partitions. 
Each candidate has an integer bound $K\le\Delta\Gamma^q$.
For at least one candidate, there is a feasible aligned schedule $O'$ such that
\begin{enumerate}[label=(\alph*),leftmargin=*]
  \item $|O'|\ge(1-2\eta)\OPT$;
  \item every execution lies in an eligible block;
  \item each block contains at most $K$ executions across all $m$ machines;
  \item $|O'|\ge m\Delta|\B|/\eta$ and $|O'|\ge K|\Scal|/\eta^8$.
\end{enumerate}
If $\OPT<H$, the instance can be solved exactly in either of the running times
$n^{O(H)}\poly(N)$ or $(2m)^{O(H)}\poly(n,T,N)$.
\end{lemma}

The implementation tests the relevant cardinalities, 
so it does not require prior knowledge of $\OPT$.
With $\Delta=1$, both $H$ and $K$ depend only on $m,\eta$.
With $\Delta=O_\eta(\log(T+1))$, both are $O_{m,\eta}(\log(T+1))$.
The latter choice makes the second exact routine pseudopolynomial for fixed $m,\eta$, 
and leads to the second implementation regime.
The density bounds are aggregate bounds: they allow individual blocks and
superblocks to contain few or no reference executions.
They control losses proportional to the number of blocks or superblocks in the later rounding analysis.

\subsection{Configurations and placement mass}

A \emph{configuration} in a block $B$ is a set of at most $K$ jobs eligible in $B$ that 
admits a feasible schedule on the $m$ machines entirely inside $B$.
In particular, each job $j$ must execute within $B\cap[r_j,d_j)$.
Let $\Ccal(B)$ be the family of these sets, including the empty set.
For every configuration used by the algorithm, 
retain a feasible schedule witnessing its membership in $\Ccal(B)$.

Introduce a nonnegative variable $x_{B,C}$ for every $C\in\Ccal(B)$.
We impose $\sum_{C\in\Ccal(B)}x_{B,C}=1$ for every $B\in\B$,
so $x_{B,C}$ is the probability of choosing configuration $C$ in block $B$.
The \emph{placement marginal} of job $j$ at $B$ and its total \emph{fractional mass} are
\[
  y_{j,B}:=\sum_{\substack{C\in\Ccal(B)\\j\in C}}x_{B,C},
  \qquad y_j:=\sum_{B\in\B}y_{j,B}.
\]
Ineligible coordinates are zero. All relaxations impose $y_j\le1$ for every job.
The quantity $y_j$ counts expected appearances across blocks; 
this inequality alone does not prevent two blocks from selecting the same job.
Write
\[
  F:=\sum_j y_j=\sum_{B\in\B}\sum_{C\in\Ccal(B)}|C|x_{B,C}
\]
for the \emph{fractional throughput}. We also impose
\[
  F\ge\frac{m\Delta|\B|}{\eta},
  \qquad F\ge\frac{K|\Scal|}{\eta^8}.
\]
For the candidate certified by Lemma~\ref{lem:partition}, 
the configurations of $O'$ give a feasible solution with $F=|O'|$.
For example, the second inequality ensures $5K|\Scal|/\eta^6\le5\eta^2F$, 
which controls the additive error in the interior matching bound of Lemma~\ref{lem:interior}.
Candidates with an infeasible strengthened relaxation are discarded;
we verify below that the reference schedule remains feasible after strengthening.

For a global job, separate its \emph{boundary mass} from its \emph{interior mass}:
\[
  a_j:=y_{j,B_L(j)},\qquad b_j:=y_{j,B_R(j)},\qquad
  z_{j,S}:=\sum_{\substack{B\subseteq S\\
                          B\notin\{B_L(j),B_R(j)\}}}y_{j,B},
  \qquad z_j:=\sum_{S\in\Scal}z_{j,S}.
\]
Eligibility implies that $z_{j,S}>0$ only if $j$ spans $S$.
Boundary blocks are excluded from $z_{j,S}$ even when their superblock is
fully spanned, so no placement is counted twice and
$a_j+b_j+z_j=y_j\le1$.
Define the aggregate quantities
\[
  L:=\sum_{j\text{ local}}y_j,\qquad
  G:=\sum_{j\text{ global}}y_j,\qquad
  Z:=\sum_{j\text{ global}}z_j,\qquad
  W:=\sum_{j\text{ global}}a_jb_j.
\]
Thus $F=L+G$.

If block configurations are sampled independently, 
$a_jb_j$ is the probability of a \emph{boundary collision}: 
job $j$ is selected at both boundary blocks.
Since $a_jb_j\le(a_j+b_j)/4\le y_j/4$, we have $W\le G/4$.
For each global job, define the \emph{boundary deficit}
\[
  e_j:=\frac{a_j+b_j}{4}-a_jb_j
      =\frac{(a_j-b_j)^2+(a_j+b_j)(1-a_j-b_j)}{4}\ge0.
\]
The aggregate deficit is
\begin{equation}\label{eq:deficit}
  \mathcal D:=\frac G4-W=\frac Z4+\sum_{j\text{ global}} e_j.
\end{equation}
This identity separates the saving relative to the collision bound $G/4$
into interior mass and the contribution of the boundary marginals.
The displayed formula for $e_j$ accounts for both unequal boundary marginals
and total boundary mass strictly between zero and one.
For a job with $y_j>0$, its contribution $z_j/4+e_j$ vanishes exactly 
when $a_j=b_j=1/2$ and $z_j=0$.
The later analysis uses this deficit to bound the loss from restricting jobs
with a small boundary marginal to one boundary.

\subsection{Consistency on triples and noncrossing boundaries}

Form a graph whose vertices are the blocks $B_1,\ldots,B_{|\B|}$ in time order.
For each global job $j$, add an edge labelled by $j$ between $B_L(j)$ and $B_R(j)$. 
This is the job's \emph{boundary edge}; it represents its two boundary placements. 
Distinct jobs with the same boundary pair give parallel edges, and local jobs give no edge.

Two boundary edges \emph{cross} if their endpoints strictly interleave.
After possibly exchanging the jobs, 
their boundary pairs then have the form $(B_u,B_v)$ and $(B_s,B_t)$ with $u<s<v<t$.
A family is \emph{noncrossing} if it contains no such pair.
Nested edges, shared endpoints, and parallel edges are allowed.
Figure~\ref{fig:boundary-graph} gives examples.
We will coordinate the configurations at the endpoints of a selected noncrossing family.

\begin{figure}[htbp]
\centering
\begin{tikzpicture}[x=1.05cm,y=.9cm,font=\small]
  \foreach \off/\paneltitle in {0/Crossing,4.5/Nested,9/Shared endpoint}{
    \draw[gray!50] (\off,0)--(\off+3,0);
    \foreach \i in {1,...,4}{
      \fill (\off+\i-1,0) circle (1.7pt);
      \node[below] at (\off+\i-1,0) {$B_{\i}$};
    }
    \node[above] at (\off+1.5,1.48) {\paneltitle};
  }
  \draw[blue!70!black,thick]
    (0,0) .. controls (0,1.15) and (2,1.15) .. (2,0);
  \draw[orange!85!black,thick]
    (1,0) .. controls (1,1.15) and (3,1.15) .. (3,0);
  \node[blue!70!black] at (.45,.83) {$j$};
  \node[orange!85!black] at (2.55,.83) {$k$};
  \draw[blue!70!black,thick]
    (4.5,0) .. controls (4.5,1.65) and (7.5,1.65) .. (7.5,0);
  \draw[orange!85!black,thick]
    (5.5,0) .. controls (5.5,.72) and (6.5,.72) .. (6.5,0);
  \node[blue!70!black,fill=white,inner sep=1pt] at (5.15,1.06) {$j$};
  \node[orange!85!black] at (6,.74) {$k$};
  \draw[blue!70!black,thick]
    (9,0) .. controls (9,1.15) and (11,1.15) .. (11,0);
  \draw[orange!85!black,thick]
    (11,0) .. controls (11,.72) and (12,.72) .. (12,0);
  \node[blue!70!black] at (10,1.04) {$j$};
  \node[orange!85!black] at (11.5,.74) {$k$};
\end{tikzpicture}
\caption{Three separate examples of boundary edges, labelled by their jobs.
Edges describe boundary pairs, not execution intervals.}
\label{fig:boundary-graph}
\end{figure}
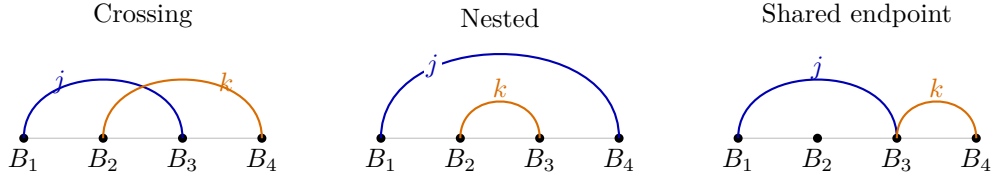

We first describe \emph{joint distributions} on complete configurations.
For each nonempty set $U\subseteq\B$ with $|U|\le3$, define
\[
  \Omega_U:=
  \left\{(C_B)_{B\in U}:
    C_B\in\Ccal(B),\
    C_B\cap C_{B'}=\varnothing\text{ for distinct }B,B'\in U
  \right\}.
\]
A tuple in $\Omega_U$ chooses one configuration per block of $U$, 
without selecting any job twice within that tuple.
The variable $\pi_U(\boldsymbol C)\ge0$ is the probability assigned to this tuple;
tuples outside $\Omega_U$ have probability zero.
For $V\subseteq U$, write $\boldsymbol C|_V$ for the tuple obtained by 
keeping only the configurations at blocks in $V$.
The \emph{marginal distribution} on $V$ sums the probabilities of all tuples
that make the same choices on $V$, regardless of their choices outside $V$.
The distributions are \emph{consistent} when this marginal equals $\pi_V$.
Formally, impose
\begin{align*}
  \sum_{\boldsymbol C\in\Omega_U}\pi_U(\boldsymbol C)&=1,
    &&\varnothing\ne U\subseteq\B,\quad |U|\le3,\\
  \sum_{\substack{\boldsymbol C\in\Omega_U\\
                  \boldsymbol C|_V=\boldsymbol D}}
       \pi_U(\boldsymbol C)&=\pi_V(\boldsymbol D),
    &&\varnothing\ne V\subsetneq U,\quad |U|\le3,\quad
      \boldsymbol D\in\Omega_V,\\
  \pi_{\{B\}}(C)&=x_{B,C},
    &&B\in\B,\quad C\in\Ccal(B).
\end{align*}
The last equation preserves the configuration distribution at each block.
The middle equation also requires agreement on every shared pair of blocks,
not just on the blocks separately.
This is a local consistency strengthening of the configuration LP
\citep{SheraliAdams1990,Laurent2003}.

For example, consider $U=\{B_1,B_2,B_3\}$ and $V=\{B_1,B_3,B_4\}$.
Let $\ell_1$ and $\ell_3$ be local jobs feasible in $B_1$ and $B_3$, respectively,
and let the other two blocks use empty configurations.
Without the pair-consistency equations, 
two distributions could induce the following probabilities on their shared blocks:
\begin{center}
\small
\begin{tabular}{@{}cccc@{}}
\toprule
$C_{B_1}$ & $C_{B_3}$ & From the distribution on $U$
                       & From the distribution on $V$\\
\midrule
$\varnothing$ & $\varnothing$ & $1/2$ & $0$\\
$\{\ell_1\}$  & $\varnothing$ & $0$   & $1/2$\\
$\varnothing$ & $\{\ell_3\}$  & $0$   & $1/2$\\
$\{\ell_1\}$  & $\{\ell_3\}$  & $1/2$ & $0$\\
\bottomrule
\end{tabular}
\end{center}
Both distributions select each local job with probability $1/2$, 
but they disagree on whether the two jobs are selected together.
They therefore cannot describe the same choices at $B_1,B_3$.
Requiring a common distribution $\pi_{\{B_1,B_3\}}$ excludes this inconsistency.

The reference schedule $O'$ satisfies all these constraints: 
in each block, give probability one to the configuration used by $O'$; for every $U$, 
give probability one to the corresponding tuple.
No job appears twice in $O'$, so each tuple belongs to $\Omega_U$ and 
its restrictions agree with the tuples used on smaller block sets.
Thus the strengthened relaxation still contains 
a feasible solution of value $|O'|$ satisfying the density bounds.

For \emph{projected consistency}, we record only selected job identities.
At each block $B$, fix a set of \emph{tracked jobs}
\[
  H_B\subseteq
  \{j:\ j\text{ is global and }B\in\{B_L(j),B_R(j)\}\}.
\]
The \emph{projected state} of a configuration is
$\Psi_B(C):=C\cap H_B$.
For instance, if $H_B=\{j\}$, feasible configurations $\{j\}$ and $\{j,k\}$
with $k\ne j$ both have projected state $\{j\}$.
We index the states by all subsets of $H_B$, 
without testing which subsets arise from configurations.
For $1\le|U|\le3$, let
\[
  \Omega_U^H:=
  \left\{(A_B)_{B\in U}:
    A_B\subseteq H_B,\
    A_B\cap A_{B'}=\varnothing\text{ for distinct }B,B'\in U
  \right\}.
\]
In this version, $\pi_U(\boldsymbol A)\ge0$ denotes 
the probability of a tuple of recorded subsets. Write $\boldsymbol A|_V=(A_B)_{B\in V}$.
The constraints are
\begin{align*}
  \sum_{\boldsymbol A\in\Omega_U^H}\pi_U(\boldsymbol A)&=1,
    &&\varnothing\ne U\subseteq\B,\quad |U|\le3,\\
  \sum_{\substack{\boldsymbol A\in\Omega_U^H\\
                  \boldsymbol A|_V=\boldsymbol D}}
       \pi_U(\boldsymbol A)&=\pi_V(\boldsymbol D),
    &&\varnothing\ne V\subsetneq U,\quad |U|\le3,\quad
      \boldsymbol D\in\Omega_V^H,\\
  p_B(A):=\pi_{\{B\}}(A)
    &=\sum_{\substack{C\in\Ccal(B)\\\Psi_B(C)=A}}x_{B,C},
    &&B\in\B,\quad A\subseteq H_B.
\end{align*}
Thus $p_B(A)$ is the probability that block $B$ records exactly the subset $A$.
If no configuration has that state, the last sum is empty and $p_B(A)=0$.
Nonnegativity and the second equation then force every tuple containing
that state to have probability zero.
This permits indexing the projected distributions without enumerating full
configurations. Projecting the configurations of $O'$ 
gives a feasible reference solution here as well.
Throughout this definition, $x$ is one fixed collection of block configuration distributions; 
Section~\ref{sec:implementation} explains how to construct the required collections and tracked sets.

We now explain why these distributions suffice for a noncrossing family.
For $|\B|\ge3$, place the block vertices, in time order, around a convex polygon.
Draw one chord for each distinct boundary pair in the chosen family, 
keeping all job labels when several jobs share a pair.
The noncrossing chords can be completed to a \emph{triangulation}, 
a subdivision of the polygon into triangles by nonintersecting diagonals.
The added sides and diagonals are used only to organize sampling; 
they do not represent additional jobs.
Every triangle involves at most three blocks, 
so the relaxation supplies its joint distribution.
Figure~\ref{fig:approach-gluing} illustrates the triangulation and the order
of sampling. In the full version, for example, after sampling the central
triangle $\{B_1,B_3,B_5\}$, sample the configuration at $B_2$ from the distribution
on $\{B_1,B_2,B_3\}$, conditional on the configurations already chosen at $B_1$ and $B_3$.

\begin{lemma}\label{lem:gluing}
Fix a set of global jobs whose boundary edges are noncrossing.
The full consistency constraints permit jointly sampling one configuration $C_B$ per block, with
\[
  \Pr[C_B=C]=x_{B,C},
  \qquad B\in\B,\ C\in\Ccal(B),
\]
such that no chosen job $j$ belongs to both $C_{B_L(j)}$ and $C_{B_R(j)}$.
Under projected consistency, the same conclusion holds if 
every chosen job satisfies $j\in H_{B_L(j)}\cap H_{B_R(j)}$.
\end{lemma}

\begin{proof}
For $|\B|\ge3$, use a triangulation containing the chosen boundary edges.
Form a second graph with one vertex per triangle and 
an edge between triangles sharing an internal diagonal. 
This graph is connected and has $|\B|-2$ vertices and $|\B|-3$ edges, so it is a tree. 
The triangles containing any fixed block vertex occur consecutively around that vertex 
and hence form a connected path in this tree.
Consequently, when a triangle is visited after its parent, 
its previously assigned blocks are exactly the two blocks it shares with that parent.
For fewer than three blocks, use a single tree node containing all blocks.

Choose a root node and visit parents before children.
The following argument uses tuples of configurations in the full version
and tuples of projected subsets in the projected version.
Sample the root tuple from its prescribed distribution.
For a child triangle with block set $U$ and parent with block set $V$, 
put $I=U\cap V$ and let $\boldsymbol d$ be the choices already fixed on $I$.
Keep those choices and sample an extending tuple $\boldsymbol c$ with probability
\[
  \frac{\pi_U(\boldsymbol c)}{\pi_I(\boldsymbol d)}
  \qquad\text{for }\boldsymbol c|_I=\boldsymbol d.
\]
Consistency makes these conditional probabilities sum to one.
Inductively, the shared pair has distribution $\pi_I$, 
so a pair with $\pi_I(\boldsymbol d)=0$ is never reached.
For every tuple with positive denominator, its probability on the child is
\[
  \pi_I(\boldsymbol d)\,
  \frac{\pi_U(\boldsymbol c)}{\pi_I(\boldsymbol d)}
  =\pi_U(\boldsymbol c).
\]
Thus every triangle used in the construction has its prescribed distribution,
and each block receives one consistent choice.
This is the standard conditional construction for distributions arranged
on a tree~\citep{LauritzenSpiegelhalter1988}.

In the full version, these choices are already configurations.
In the projected version, after choosing all subsets, 
sample independently at each block conditional on its chosen subset:
for $p_B(A)>0$, use
\[
  \Pr[C_B=C\mid A_B=A]=
  \begin{cases}
    x_{B,C}/p_B(A),&\Psi_B(C)=A,\\
    0,&\text{otherwise}.
  \end{cases}
\]
Only subsets with positive $p_B(A)$ can be reached.
For any configuration $C$ of positive mass, its unconditional probability is
$p_B(\Psi_B(C))x_{B,C}/p_B(\Psi_B(C))=x_{B,C}$.
This reconstruction uses the same $x$ as the projected constraints.

The endpoints of every chosen boundary edge belong to one sampled block set:
a triangle when $|\B|\ge3$, and the single set otherwise.
In the full version, disjointness of its configurations prevents the job from appearing at both endpoints.
In the projected version, the job is tracked at both endpoints, 
so disjointness of the recorded subsets gives the same conclusion after reconstruction.
The construction uses $O(|\B|)$ draws from finite distributions.
\end{proof}

The relaxation is solved before the noncrossing family is selected.
Including distributions on all sets of at most three blocks therefore 
lets the rounding use any triangulation containing the selected family.
The construction preserves each block's configuration distribution and 
the distributions on the triangles it uses: 
these concern full configurations in the full version and recorded subsets in the projected version.
It need not preserve the distributions assigned to other pairs or triples.
For example, no triangle in Figure~\ref{fig:approach-gluing} contains both $B_2$ and $B_4$. 
The construction need not preserve the prescribed joint distribution $\pi_{\{B_2,B_4\}}$.

The sampled configurations need not yet form a feasible schedule, 
because jobs can still appear more than once outside the coordinated boundary pairs.
Section~\ref{sec:rounding} obtains a feasible schedule by deleting global interior executions, 
restricting every global job outside the coordinated family to one fixed boundary, and retaining local executions.
The lemma supplies the remaining guarantee: each job allowed at both boundaries is selected at most once.

\section{Rounding the configuration solution}\label{sec:rounding}

Fix a feasible solution of the relaxation in Section~\ref{sec:prelim}, 
including the configuration probabilities $x_{B,C}$ and their stored feasible schedules.
Write $x_B=(x_{B,C})_{C\in\Ccal(B)}$ for the distribution at block $B$.
We use either full consistent distributions or projected consistent distributions
linked to this same $x$, with the supplied tracked sets $H_B$.
Both rounding procedures therefore have the same placement marginals $y_{j,B}$
and the same quantities $F,L,G,Z,W$, and $\mathcal D$.

A global job is \emph{good} if $a_j,b_j\ge1/4$, and \emph{bad} otherwise.
For projected consistency, define
\[
  J_{\rm miss}:=
  \{j\text{ good}:j\notin H_{B_L(j)}\cap H_{B_R(j)}\}.
\]
These jobs are not tracked at both boundaries.
For full consistency, set $J_{\rm miss}=\varnothing$.

The first samples block configurations independently and 
computes a maximum job--slot matching. 
The second coordinates selected boundary placements.
Let $A_\times$ and $A_{\rm keep}$ be the numbers of jobs in their respective schedules. 
Algorithm~\ref{alg:rounding} summarizes the two constructions and returns the larger schedule.

\begin{algorithm}[htbp]
\caption{Rounding a consistent configuration solution}
\label{alg:rounding}
\begin{algorithmic}[1]
\Statex \textbf{Input:} A feasible configuration solution from
Section~\ref{sec:prelim}, with probabilities $x_{B,C}$ and stored schedules;
full consistent distributions, or projected consistent distributions linked to this same $x$ with fixed tracked sets $H_B$.
\Statex \textbf{Output:} A feasible schedule.
\State Compute $a_j,b_j$, the good and bad jobs, and $J_{\rm miss}$; 
use $J_{\rm miss}=\varnothing$ in the full version.
\State Construct the crossing graph on all good jobs and choose a maximal matching $\mathcal X$.
\State Independently sample a configuration in each block from $x_B$.
Use the stored schedules to form slots and compute a maximum job--slot
matching. Retain its schedule and cardinality $A_\times$.
\State Apply Lemma~\ref{lem:gluing} to jointly sample block configurations,
coordinating the good jobs that are unmatched by $\mathcal X$ and outside $J_{\rm miss}$. 
In the projected version, reconstruct full configurations 
conditionally on the sampled states as specified in that lemma.
\State In this second sample, retain all local executions and delete all global interior executions.
\State For each bad job, matched good job, or job in $J_{\rm miss}$, 
retain its sampled execution only at its more likely boundary, if present, with fixed tie-breaking. 
For each remaining good job, retain its boundary execution if present. 
Let $A_{\rm keep}$ be the resulting cardinality.
\State \Return the larger of the schedules of sizes $A_\times$ and $A_{\rm keep}$.
\end{algorithmic}
\end{algorithm}

\subsection{Independent sampling and interior matching}

Independently for each block $B$, choose a configuration $C_B$ with probability $x_{B,C_B}$ and use its stored schedule.
A \emph{slot} is one execution interval in this schedule, together with its machine index. 
We also record the job that produced the slot.
Form a bipartite graph with all jobs on one side and the sampled slots on the other. 
Job $j$ is adjacent to a slot $[u,v)$ exactly when $\max\{r_j,u\}+p_j\le\min\{d_j,v\}$.
Thus an adjacent job fits in that slot within its release time and deadline.
A \emph{maximum job--slot matching} gives the first schedule: 
execute each matched job on its slot's machine, starting at $\max\{r_j,u\}$.
Slots on the same machine are disjoint, and each job is matched at most once,
so this schedule is feasible and has cardinality $A_\times$.
The graph includes every compatible job--slot pair, 
even when the placement does not satisfy the additional eligibility restriction of the relaxation.

A slot produced by a global job is a \emph{boundary slot} 
if its block is one of that job's boundary blocks, and an \emph{interior slot} otherwise.
This classification refers to the job that produced the slot, 
including when the slot is reassigned to another job.
Every interior slot lies in a superblock spanned by the job that produced it.
A boundary slot remains excluded from this category even if its job spans the entire superblock.

For the analysis, first form an integral matching that 
assigns each sampled local job to its own slot and 
each global job appearing at a boundary to one of its own boundary slots.
Independent sampling selects a global job at least once at its boundaries with probability $a_j+b_j-a_jb_j$.
The expected size of this initial matching is therefore
\[
  L+\sum_{j\text{ global}}(a_j+b_j-a_jb_j)=F-Z-W.
\]
It uses no interior slots. Write $D_j:=(1-a_j)(1-b_j)$
for the probability that a global job is \emph{absent at both boundaries},
meaning $j\notin C_{B_L(j)}$ and $j\notin C_{B_R(j)}$.
Absence here refers to the sampled configurations, before the matching is run.

A \emph{fractional matching} assigns nonnegative numbers $f_{j,\sigma}$ to job--slot pairs, 
with $f_{j,\sigma}=0$ for incompatible pairs, subject to
\[
  \sum_\sigma f_{j,\sigma}\le1\quad\text{for every job }j,
  \qquad
  \sum_j f_{j,\sigma}\le1\quad\text{for every slot }\sigma.
\]
Define $f_j:=\sum_\sigma f_{j,\sigma}$, the total fractional assignment to job $j$, 
and $|f|:=\sum_j f_j$, the value of the matching.
These numbers are used only to lower-bound the maximum integral matching.
We write $\mathbf1\{E\}$ for the indicator of an event $E$.

The next lemma recovers interior mass using jobs absent at both boundaries.
Its bound on each $f_j$ will control the amount that a later crossing exchange must replace.

\begin{lemma}\label{lem:interior}
For every outcome of independent block sampling, 
there is a fractional matching $f$ using only interior slots and assigning positive mass only to global jobs.
For every global job $j$,
\[
  f_j\le\frac{z_j}{D_j}
       \mathbf1\{j\text{ is absent at both boundaries}\},
\]
where the right-hand side is defined to be zero when $D_j=0$.
Moreover,
\[
  \E|f|\ge(1-6\eta)Z-\frac{5K|\Scal|}{\eta^6}.
\]
The union of $f$ and the initial integral matching is a feasible fractional matching.
\end{lemma}

\begin{proof}
We adapt the \emph{size grouping} argument of \citep{ArmbrusterGrandoniTinguelyWiese2026}, 
retaining the stated bound on the assignment to each job and excluding every boundary slot.
First observe that $z_j\le1-a_j-b_j\le(1-a_j)(1-b_j)=D_j$.
In particular, $D_j=0$ implies $z_j=0$, and $z_j/D_j\le1$ whenever $D_j>0$.

Fix a superblock $S$ and put $Z_S:=\sum_{j\text{ global}}z_{j,S}$.
If $Z_S=0$, use the empty matching in $S$.
Otherwise, order the jobs with $z_{j,S}>0$ by nonincreasing processing time.
Each of these jobs spans $S$.
For $\ell=1,\ldots,q-1$, take the longest prefix whose total $z_{j,S}$ mass is at most $\ell\eta Z_S$.
Starting from the empty prefix, take the differences of consecutive prefixes
and then the final suffix to form groups $J_1,\ldots,J_q$.
No job is split between groups. Since $z_{j,S}\le1$, their masses satisfy
\[
  n_\ell:=\sum_{j\in J_\ell}z_{j,S}
  \in[\eta Z_S-1,\eta Z_S+1].
\]
Empty prefixes and groups are allowed.
Thus the groups have approximately equal placement mass, 
while every job in an earlier group is at least as long as every job in a later group.

Let $\overline N_\ell$ be the number of sampled interior slots inside $S$ produced by jobs in $J_\ell$.
The definition of $z_{j,S}$ gives $\E\overline N_\ell=n_\ell$.
This count is a sum of independent block contributions, each in $[0,K]$.
Consequently, $\operatorname{Var}(\overline N_\ell)\le K n_\ell$.
For every global job, define its allowed fractional assignment within $S$ by
\[
  Y_{j,S}:=\frac{z_{j,S}}{D_j}
       \mathbf1\{j\text{ is absent at both boundaries}\},
\]
again setting the expression to zero when $D_j=0$.
Then $0\le Y_{j,S}\le1$, $\E Y_{j,S}=z_{j,S}$, and
\[
  \sum_{S\in\Scal}Y_{j,S}
  =\frac{z_j}{D_j}\mathbf1\{j\text{ is absent at both boundaries}\}\le1
\]
under the same zero convention.
The factor $1/D_j$ compensates for using $j$ only when it is absent at both boundaries, 
so its expected allowed assignment remains $z_{j,S}$.

For each $\ell=2,\ldots,q$, assign jobs in $J_\ell$ to interior slots produced by jobs in $J_{\ell-1}$.
Every such pair is compatible: the job to be assigned spans $S$ 
and is no longer than the job that produced the slot.
Writing $N_\ell:=\sum_{j\in J_\ell}Y_{j,S}$, we can therefore construct a fractional matching of value
$\min\{N_\ell,\overline N_{\ell-1}\}$
that assigns at most $Y_{j,S}$ to each job $j$.
To see this, fill the unit capacities of the available slots successively with the allowed fractional assignments; 
every job can use every slot in this pairing.
Different pairings assign disjoint groups of jobs and use disjoint sets of slots.
Jobs in $J_1$ receive no assignment, removing at most $n_1\le\eta Z_S+1$ from the expected allowed total; 
slots produced by $J_q$ are unused.
Any further loss comes from a mismatch between available slots and allowed assignments in the paired groups.
Figure~\ref{fig:size-groups} illustrates the construction.

\begin{figure}[htbp]
\centering
\begin{tikzpicture}[x=1cm,y=1cm,>=Stealth,font=\small,
  group/.style={draw,rounded corners=2pt,minimum width=2.7cm,
                minimum height=.8cm,align=center}]
  \node at (3.6,2.15) {Processing times do not increase from left to right};
  \node at (0,1.65) {$J_{\ell-1}$};
  \node at (3.6,1.65) {$J_\ell$};
  \node at (7.2,1.65) {$J_{\ell+1}$};
  \node[group,fill=green!8] (s0) at (0,.9)
    {available slots\\$\overline N_{\ell-1}$};
  \node[group,fill=green!8] (s1) at (3.6,.9)
    {available slots\\$\overline N_\ell$};
  \node[group,fill=green!8] (s2) at (7.2,.9)
    {available slots\\$\overline N_{\ell+1}$};
  \node[group,fill=blue!6] (j0) at (0,-1.4)
    {assignment limit\\$N_{\ell-1}$};
  \node[group,fill=blue!6] (j1) at (3.6,-1.4)
    {assignment limit\\$N_\ell$};
  \node[group,fill=blue!6] (j2) at (7.2,-1.4)
    {assignment limit\\$N_{\ell+1}$};
  \draw[->,thick,blue!70!black] (s0.south)--node[above,sloped]
    {every job fits} (j1.north);
  \draw[->,thick,blue!70!black] (s1.south)--node[above,sloped]
    {every job fits} (j2.north);
  \draw[->,gray,dashed] (-1.6,-.33)--(j0.north);
  \draw[->,gray,dashed] (s2.south)--(8.8,-.17);
\end{tikzpicture}
\caption{Three consecutive size groups within a superblock $S$, where $3\le\ell\le q-2$. 
A group's interior slots serve jobs in the next group.
All these jobs span $S$, and the earlier group's slots are long enough for the later group's jobs. 
The top quantities are integer slot counts; 
the bottom quantities are total allowed fractional assignments.
Dashed arrows continue the adjacent pairings outside the picture.}
\label{fig:size-groups}
\end{figure}
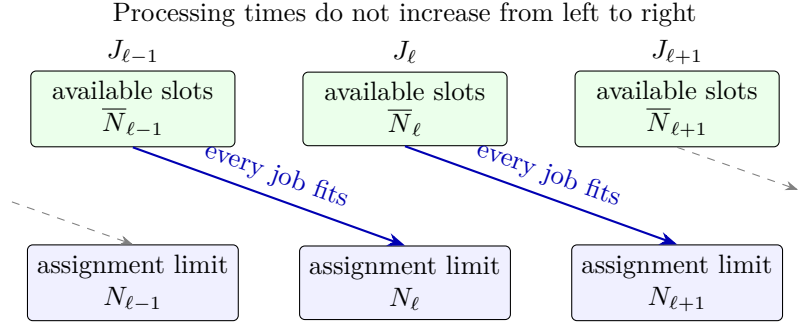

It remains to bound the expected value of these pairings.
The variables $Y_{j,S}$ can be dependent because different jobs may share boundary blocks. 
For this estimate only, omit the absence test at a boundary whose placement marginal is below $\eta$.
For $D_j>0$, the resulting variable is
\[
  \widetilde Y_{j,S}:=
  \frac{z_{j,S}}{D_j}
  \prod_{\substack{B\in\{B_L(j),B_R(j)\}\\y_{j,B}\ge\eta}}
       \mathbf1\{j\notin C_B\}.
\]
An empty product equals one, and $\widetilde Y_{j,S}=0$ when $D_j=0$.
These modified variables are only analytical upper bounds; 
the actual fractional matching still uses $Y_{j,S}$.
Removing an absence test increases the variable and divides its expectation by a factor greater than $1-\eta$.
At most two tests are removed, so
\[
  Y_{j,S}\le\widetilde Y_{j,S}\le1,
  \qquad
  z_{j,S}\le\E\widetilde Y_{j,S}
  \le(1-\eta)^{-2}z_{j,S}\le(1+3\eta)z_{j,S}.
\]

Each modified variable depends on at most two independent block samples.
At any block $B$, at most $K/\eta$ variables retain its absence test:
each corresponding job has $y_{j,B}\ge\eta$, whereas
\[
  \sum_j y_{j,B}=
  \sum_{C\in\Ccal(B)}|C|x_{B,C}\le K.
\]
Each modified variable therefore shares a retained block with at most $2K/\eta$ other variables. 
Variables with no shared retained block are independent, 
so only these pairs can have nonzero covariance.
For each such pair, use
$
  2|\operatorname{Cov}(\widetilde Y_{i,S},\widetilde Y_{j,S})|
  \le\operatorname{Var}(\widetilde Y_{i,S})
     +\operatorname{Var}(\widetilde Y_{j,S}).
$
After expanding the variance and applying this bound, 
each individual variance is counted at most $1+2K/\eta$ times.
Since a $[0,1]$ variable has variance at most its mean, 
the sum $\widetilde N_\ell:=\sum_{j\in J_\ell}\widetilde Y_{j,S}$ satisfies
\[
  \operatorname{Var}(\widetilde N_\ell)
  \le(1+2K/\eta)(1+3\eta)n_\ell
  \le\frac{4K n_\ell}{\eta}.
\]

Each of $\widetilde N_\ell$ and $\overline N_\ell$ is at least
$(1-\eta)n_\ell-\frac{4K}{\eta^5}$ with probability at least $1-\eta^2$.
If $n_\ell<4K/\eta^5$, the threshold is negative and the assertion is immediate.
Otherwise, each mean is at least $n_\ell$, and Chebyshev's inequality bounds
the failure probability for $\widetilde N_\ell$ by
\[
  \frac{4K n_\ell/\eta}{(\eta n_\ell)^2}
  =\frac{4K}{\eta^3 n_\ell}\le\eta^2.
\]
The variance bound for $\overline N_\ell$ gives an even smaller probability.

Set $c:=(1-\eta)(\eta Z_S-1)-\frac{4K}{\eta^5}$.
If $c>0$, a union bound implies that both $\widetilde N_\ell$ and $\overline N_{\ell-1}$ are 
at least $c$ with probability at least $1-2\eta^2$.
No independence between these two quantities is required.
For $c\le0$, their minimum is nonnegative, so in either case
$
  \E \left[ \min\{\widetilde N_\ell,\overline N_{\ell-1}\} \right]
  \ge(1-2\eta^2)c.
$
Replacing $\widetilde N_\ell$ by $N_\ell$ decreases the minimum by at most $\widetilde N_\ell-N_\ell$, 
whose expectation is at most $3\eta n_\ell$.
Thus
\[
  \E \left[ \min\{N_\ell,\overline N_{\ell-1}\} \right]
  \ge(1-2\eta^2)c-3\eta n_\ell.
\]
Summing over $\ell=2,\ldots,q$ gives a fractional matching $f_S$ within $S$
with
\begin{align*}
  \E|f_S|
  &\ge(1-\eta)^2(1-2\eta^2)Z_S-3\eta Z_S
       -\frac1\eta-\frac{4K}{\eta^6}
  \ge(1-6\eta)Z_S-\frac{5K}{\eta^6}.
\end{align*}
Here we used $q\eta=1$, $\sum_{\ell=2}^q n_\ell\le Z_S$, and
$(1-\eta)^2(1-2\eta^2)\ge1-2\eta-2\eta^2$.
The last inequality follows from $2\eta^2\le\eta$ and $K\ge1$.

Finally, combine the matchings over all superblocks.
Their slots are disjoint. A job may receive assignments in several superblocks, 
but their sum is at most $\sum_S Y_{j,S}$, 
giving the stated bound on $f_j$ and the unit capacity of every job.
Summing the expectation bounds proves the claimed value.
Only global jobs absent at both boundaries receive assignments in $f$, 
so they are distinct from the jobs in the initial matching.
The two matchings also use disjoint slots, 
because $f$ excludes every boundary slot and every slot produced by a local job.
Their union is therefore feasible.
\end{proof}

\subsection{Recovering jobs from crossing boundaries}

Consider global jobs $j,k$ with boundary pairs $(B_u,B_v)$ and $(B_s,B_t)$, where $u<s<v<t$.
Then $B_s\subseteq[r_j,d_j)$ and $B_v\subseteq[r_k,d_k)$.
If $p_j\le p_k$, job $j$ fits in every slot produced by $k$ in $B_s$.
When $j$ is absent at both boundaries and $k$ appears at both, 
we may retain $k$ in $B_t$ and assign $j$ to $k$'s slot in $B_s$.
If $p_k<p_j$, exchange their roles, using $j$'s slot in $B_v$ for $k$ and retaining $j$ in $B_u$.
Figure~\ref{fig:approach-crossing} illustrates this \emph{crossing exchange}.
The new placement satisfies the original scheduling constraints; 
it need not satisfy the relaxation's additional eligibility restriction.

\begin{lemma}\label{lem:rescue}
Let $\mathcal X$ be a collection of crossing pairs with no job in two pairs,
fixed before the configurations are sampled. Then
\[
  \E \left[ A_\times \right] 
  \ge F-W+\sum_{\{j,k\}\in\mathcal X}(a_jb_j)(a_kb_k)-8\eta F.
\]
\end{lemma}

\begin{proof}
Combine the initial integral matching with the fractional matching $f$ from Lemma~\ref{lem:interior}.
For each pair, let $r$ be its shorter job and $d$ its longer job, 
breaking ties by a fixed job order.
The crossing identifies a boundary slot of $d$ that can accommodate $r$.
If $r$ is absent at both boundaries and $d$ appears at both, 
retain $d$ at its other boundary. 
Move its initial assignment there if necessary, 
thereby freeing the slot to be used by $r$.
On this event, $f_d=0$ by Lemma~\ref{lem:interior}, 
because $d$ is not absent at both boundaries. 
Relocating $d$ therefore preserves its job capacity and the matching value.
Next delete all assignments $f_{r,\sigma}$ to $r$ and assign one unit of $r$ to the freed slot.
This deletes fractional assignments to $r$, not slots produced by $r$.
For $D_r>0$, the increase in matching value is at least
$1-f_r\ge1-\frac{z_r}{D_r}\ge0$.
If $D_r=0$, the exchange event has probability zero.

All these exchanges can be performed simultaneously.
The pairs use distinct job identities, so no job receives two new assignments
or is relocated by two exchanges.
Their freed slots are distinct, even if two pairs share a block.
Every freed slot is a boundary slot of the job that produced it, 
and hence is unused by $f$.
Deleting fractional assignments only releases capacity.
Thus the exchanges preserve every job and slot capacity.

The four boundary blocks of a crossing pair are distinct and sampled independently. 
The exchange event therefore has probability $D_ra_db_d$.
For $D_r>0$, the bound on $f_r$ holds for every outcome, including this event,
so the expected gain from the pair is at least
\begin{align*}
  D_ra_db_d\left(1-\frac{z_r}{D_r}\right)
  &=a_db_d(D_r-z_r)
  =a_db_d(1-y_r+a_rb_r)
  \ge(a_rb_r)(a_db_d).
\end{align*}
When $D_r=0$, both the event probability and the claimed product are zero.
This argument does not require independence between $f_r$ and the event,
or between the events for different pairs.

Before the exchanges, the expected matching value is at least
\[
  (F-Z-W)+(1-6\eta)Z-\frac{5K|\Scal|}{\eta^6}
  =F-W-6\eta Z-\frac{5K|\Scal|}{\eta^6}.
\]
Linearity of expectation adds the gains from all pairs.
The density bound and $Z\le F$ give
\[
  6\eta Z+\frac{5K|\Scal|}{\eta^6}
  \le(6\eta+5\eta^2)F\le8\eta F.
\]
Finally, the fractional matching polytope of a bipartite graph is integral.
For every outcome, the maximum matching computed by the algorithm is at least
as large as the fractional matching just constructed.
This proves the claimed inequality.
\end{proof}

Deleting the old assignments to $r$ is essential: 
the exchange gains $1-f_r$, because $r$ may already carry fractional interior assignment.

\subsection{Coordinating a noncrossing family}

Form the \emph{crossing graph}: its vertices are good jobs, 
and two vertices are adjacent when their jobs' boundary edges cross.
Choose a \emph{maximal matching} $\mathcal X$, 
meaning a set of pairwise vertex-disjoint edges to which no further edge can be added, 
and let $\nu:=|\mathcal X|$.
This choice uses only the configuration marginals and is made before sampling.
The unmatched good jobs have noncrossing boundary edges: 
any crossing between two unmatched jobs could be added to $\mathcal X$, contradicting maximality.

Apply Lemma~\ref{lem:gluing} to the good jobs that are unmatched by $\mathcal X$
and do not belong to $J_{\rm miss}$.
Their boundary edges are noncrossing, 
and every such job is tracked at both boundaries when projected consistency is used.
The lemma jointly samples block configurations with probabilities $x_{B,C}$
and prevents each of these jobs from appearing at both boundaries.

In this sample, retain all local executions and delete every global interior execution. 
Each coordinated job may retain whichever boundary execution appears.
For every other global job, fix its boundary with the larger marginal, 
breaking ties deterministically, and retain its sampled execution only there, if present.
Thus a job outside the coordinated family is not guaranteed to survive.
All retained executions come from feasible block schedules, 
and every job appears at most once, giving a feasible schedule of cardinality $A_{\rm keep}$.

Let $P:=\sum_{j\text{ bad}}e_j$ be the total boundary deficit of bad jobs.

\begin{lemma}\label{lem:coordination}
The coordinated schedule satisfies
\[
  \E \left[ A_{\rm keep} \right]
  \ge F-Z-4P-\nu-\frac{|J_{\rm miss}|}{2}
  \ge F-4\mathcal D-\nu-\frac{|J_{\rm miss}|}{2}.
\]
\end{lemma}

\begin{proof}
Before restricting jobs to one boundary, the expected number of local and global boundary appearances is $F-Z$, 
because each block has distribution $x_B$.
Restricting a global job to its more likely boundary loses exactly $\min\{a_j,b_j\}$ from this expected number.
For a bad job, write $a,b$ for its boundary marginals and $e=(a+b)/4-ab$.
If $a\le b$, then $a<1/4$ and $4e-a=b(1-4a)\ge0$.

The other ordering is symmetric, so bad jobs lose at most $4P$ in total.
For every global job, $\min\{a_j,b_j\}\le1/2$ since $a_j+b_j\le1$.
The $2\nu$ matched good jobs therefore lose at most $\nu$, 
and the jobs in $J_{\rm miss}$ lose at most another $|J_{\rm miss}|/2$.
If a missed job is already matched, this estimate only counts its loss twice.
The remaining global jobs are coordinated, 
so they cannot be duplicated and their boundary appearances count distinct retained jobs.
This proves the first inequality.
For the second, use $P\le\sum_{j\text{ global}}e_j$ and $Z+4\sum_{j\text{ global}}e_j=4\mathcal D$.
\end{proof}

The threshold $1/4$ makes a bad job's discarded marginal affordable from its deficit. 
For good jobs, the next inequality relates the crossing gain to the same deficit.

\subsection{Combining the two schedules}

\begin{lemma}\label{lem:good-product}
For good jobs $j,k$, $(a_jb_j)(a_kb_k)\ge\frac1{16}-\frac34(e_j+e_k)$.
\end{lemma}

\begin{proof}
For $a,b\ge1/4$ with $a+b\le1$, put $s:=a+b\in[1/2,1]$ and $e:=s/4-ab$.
Then
\[
  ab+3e=\frac{3s}{4}-2ab
  \ge\frac{3s}{4}-\frac{s^2}{2}
  =\frac14+\frac12(s-1/2)(1-s)\ge\frac14.
\]
If $3(e_j+e_k)\ge1/4$, the claimed lower bound is nonpositive.
Otherwise, both $1/4-3e_j$ and $1/4-3e_k$ are positive, so
$
  (a_jb_j)(a_kb_k)
  \ge(1/4-3e_j)(1/4-3e_k)
  =\frac1{16}-\frac34(e_j+e_k)+9e_je_k.
$
\end{proof}

Let $Q:=\sum_{j\text{ matched by }\mathcal X}e_j$.
Since the matched jobs are distinct, $Q\le\sum_{j\text{ global}}e_j\le\mathcal D$.
Lemmas~\ref{lem:rescue} and~\ref{lem:good-product}, together with $F-W=L+3G/4+\mathcal D$, give
\begin{align*}
  \E \left[ A_\times \right]
  &\ge L+\frac34G+\mathcal D+\frac{\nu}{16}-\frac{3Q}{4}-8\eta F
  \ge L+\frac34G+\frac{\mathcal D}{4}+\frac{\nu}{16}-8\eta F.
\end{align*}
Lemma~\ref{lem:coordination} gives the complementary bound
\[
  \E \left[ A_{\rm keep} \right]
  \ge L+\frac34G+\frac G4-4\mathcal D-\nu-\frac{|J_{\rm miss}|}{2}.
\]
More selected crossing pairs increase the first guarantee and decrease the second. 
Their dependence on $\mathcal D$ is also opposite, 
allowing both quantities to cancel when the bounds are combined.

\begin{theorem}\label{thm:rounding}
With full consistency, Algorithm~\ref{alg:rounding} returns a feasible schedule whose expected cardinality satisfies
\[
  \E \left[  \max\{A_\times,A_{\rm keep}\} \right]
  \ge L+\frac{13}{17}G-8\eta F.
\]
The same guarantee holds with projected consistency if $|J_{\rm miss}|\le\eta F$.
\end{theorem}

\begin{proof}
For every pair of output sizes, their maximum is at least their convex combination with weights $16/17$ and $1/17$.
These weights cancel both $\mathcal D$ and $\nu$, since $16(\mathcal D/4+\nu/16)=4\mathcal D+\nu$.
Consequently,
\begin{align*}
  \E \left[ \max\{A_\times,A_{\rm keep}\}  \right]
  &\ge\frac{16}{17}\E \left[ A_\times  \right] +\frac1{17}\E \left[ A_{\rm keep}  \right]
  \ge L+\frac{13}{17}G
       -\frac{128}{17}\eta F-\frac{|J_{\rm miss}|}{34}.
\end{align*}
No independence between the two output schedules is needed.
In the full version, $J_{\rm miss}=\varnothing$.
In the projected version, the assumed bound on $|J_{\rm miss}|$ gives total error at most
\[
  \left(\frac{128}{17}+\frac1{34}\right)\eta F
  =\frac{257}{34}\eta F<8\eta F.
\]
\end{proof}

\section{Implementation and approximation guarantees}\label{sec:implementation}

The implementation in \citep{ArmbrusterGrandoniTinguelyWiese2026}
enumerates configurations when $K$ is constant and uses color coding to optimize
additive job weights when $K$ is logarithmic in $T$.
Our stronger relaxation must also enforce correlations between blocks.
The main challenge in the logarithmic case is to represent these correlations
without enumerating configurations, while keeping the tracked sets small and useful after reoptimization.
We first give the constant-$K$ implementation, then construct and solve the projected relaxation. 
Throughout, $m$ is fixed until the final reduction in Section~\ref{sec:completion}.

\subsection{The polynomial-time implementation}

Set $\Delta=1$ in Lemma~\ref{lem:partition}, so $K$ and $H$ depend only on $m,\eta$. 
For each candidate partition, enumerate all sets of at most $K$ eligible jobs, 
their assignments to machines, and their orders on each machine.
For a fixed order, schedule each job as early as its release time and the preceding completion permit, 
starting no earlier than the block's left endpoint.
The order is feasible exactly when these executions meet their deadlines and
finish within the block. This test uses binary arithmetic and 
produces the stored schedules needed for rounding.

There are $n^{O(K)}$ configurations per block. Enumerating tuples on at most
three blocks gives $n^{O(K)}\poly(|\B|)$ variables and constraints for the full consistency relaxation of Section~\ref{sec:prelim}.
Maximize $F$, retaining the job and density constraints, and discard infeasible candidates. 
For the candidate supplied by Lemma~\ref{lem:partition}, the reference schedule $O'$ is feasible in this LP. 
Thus selecting a feasible candidate with largest objective gives $F\ge|O'|\ge(1-2\eta)\OPT$.
Apply Algorithm~\ref{alg:rounding} to its configuration distributions.
The same subset and order enumeration tests target cardinalities up to $H$ before constructing the partitions, 
solving the case $\OPT<H$ exactly as described in Appendix~\ref{app:partition}.

\subsection{Projected consistency and configuration copies}

For the pseudopolynomial implementation, set
\[
  \Delta=4\eta^{-4}\bigl(\lceil\log_2(T+1)\rceil+1\bigr).
\]
This is an integer because $\eta=1/q$.
Now $K,H=O_{m,\eta}(\log(T+1))$, so explicit configuration enumeration would take $n^{O_{m,\eta}(\log(T+1))}$ time.
Instead, we use the projected consistency constraints of Section~\ref{sec:prelim}.
Their tracked sets are chosen from successive solutions; 
the objective used to select these solutions is specified in the next subsection.

For one candidate partition, write $y=(y_{j,B})_{j,B}$ for the collection of placement marginals. 
A \emph{configuration copy} is a separate collection $x^{[k]}=(x^{[k]}_{B,C})_{B,C}$ of configuration distributions representing $y$.
At stage $i\ge0$, we have copies $x^{[0]},\ldots,x^{[i]}$, each satisfying
\[
  \sum_{C\in\Ccal(B)}x^{[k]}_{B,C}=1,
  \qquad
  \sum_{\substack{C\in\Ccal(B)\\j\in C}}x^{[k]}_{B,C}=y_{j,B}
  \quad B\in\B,\ 0\le k\le i.
\]
All copies therefore have the same placement marginals, but their probabilities on complete configurations may differ. 
The job and density constraints are imposed on the common marginals $y$.

After solving stage $i$, let $y^i$ be its selected marginals and set
$F_i:=\sum_{j,B}y^i_{j,B}$. Define
\[
  H_B^i:=\left\{j:\ j\text{ is global},\
       B\in\{B_L(j),B_R(j)\},\ y^i_{j,B}\ge\frac18\right\}.
\]
Since $\sum_jy^i_{j,B}\le K$, we have $|H_B^i|\le8K$.
Add a new copy $x^{[i+1]}$ and the projected joint distributions on these tracked sets 
for every group of at most three blocks.
Impose their consistency equations from Section~\ref{sec:prelim}.
Let $p_B^{[i+1]}(A)$ denote the probability of state $A$ at block $B$ in these new distributions, 
and link it to the new copy by
\[
  p_B^{[i+1]}(A)
  =\sum_{\substack{C\in\Ccal(B)\\C\cap H_B^i=A}}x^{[i+1]}_{B,C}
  \qquad A\subseteq H_B^i.
\]
Earlier variables and constraints remain and are reoptimized; 
the tracked sets used in their constraints do not change.
Figure~\ref{fig:copies} illustrates the first two additions.

Let $P_i$ be the set of marginals $y$ that can be extended to a feasible solution of all stage-$i$ constraints. 
Then $P_{i+1}\subseteq P_i$, and these sets are convex. 
For the candidate certified by Lemma~\ref{lem:partition},
the reference schedule $O'$ remains feasible at every stage: 
use its configurations with probability one in every copy, 
and their recorded subsets with probability one in every projected distribution.
Each group of three blocks has at most $2^{24K}$ projected tuples.
Crucially, each copy is linked only to the projected constraints introduced with that copy. 
This keeps its pricing objective dependent on just one tracked subset, as shown in Section~\ref{sec:pricing}.

\begin{figure}[htbp]
\centering
\begin{tikzpicture}[font=\small,
  box/.style={draw,rounded corners=2pt,align=center,
              minimum width=3.05cm,minimum height=.72cm,inner sep=5pt}]
  \node[box,fill=gray!10,minimum width=4.1cm] (y) at (4,2.1)
    {Common placement marginals $y_{j,B}$};
  \node[box] (x0) at (0,0) {$x^{[0]}$\\base configuration copy};
  \node[box,fill=blue!7] (x1) at (4,0) {$x^{[1]}$\\configuration copy};
  \node[box,fill=green!8] (x2) at (8,0) {$x^{[2]}$\\configuration copy};
  \draw (y.south west)--(x0.north);
  \draw (y.south)--node[right,fill=white,inner sep=2pt]
    {same marginals} (x1.north);
  \draw (y.south east)--(x2.north);
  \node[box,fill=blue!7] (p1) at (4,-1.85)
    {Projected distributions\\tracked sets $H_B^0$};
  \node[box,fill=green!8] (p2) at (8,-1.85)
    {Projected distributions\\tracked sets $H_B^1$};
  \draw (x1.south)--node[right,fill=white,inner sep=2pt]
    {record $C\cap H_B^0$} (p1.north);
  \draw (x2.south)--node[right,fill=white,inner sep=2pt]
    {record $C\cap H_B^1$} (p2.north);
\end{tikzpicture}
\caption{Separate configuration copies share placement marginals.
Connections denote equality constraints.
Copies $x^{[1]}$ and $x^{[2]}$ are linked only to their respective projected distributions. 
All variable values may change on reoptimization, 
while the tracked sets attached to earlier copies remain fixed.}
\label{fig:copies}
\end{figure}
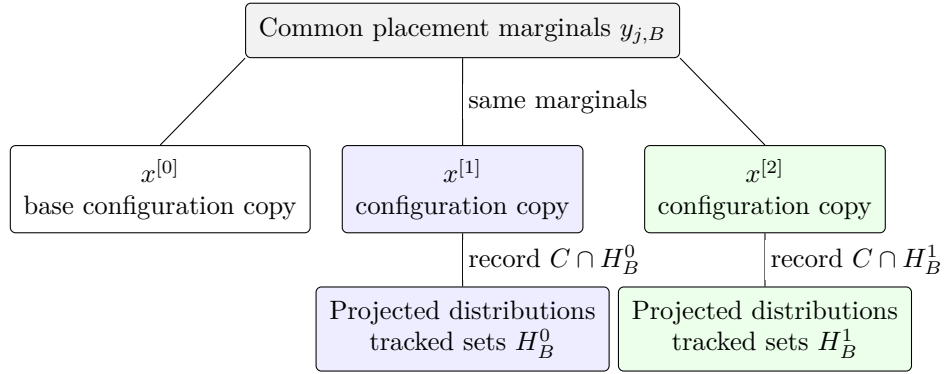

\subsection{Controlling changes in the marginals}

Adding projected constraints can change the optimal marginals and make previously untracked jobs good. 
We control these changes by subtracting a small quadratic penalty from $F$.
The penalty ensures that a large change between successive solutions forces a measurable decrease in the objective. 
A rational piecewise-linear approximation lets us obtain this property while solving only LPs.

Let $d_y:=n|\B|$, including coordinates fixed to zero, and set
\[
  \beta:=\eta,\qquad \gamma:=\eta/100,\qquad
  h:=\lceil4d_y/\gamma\rceil.
\]
For $t\in[0,1]$, define
\[
  \phi(t):=\max_{\ell=0,\ldots,h}
       \left\{2\frac\ell h t-\left(\frac\ell h\right)^2\right\},
  \qquad
  \widehat\Phi(y):=F-\beta\sum_{j,B}\phi(y_{j,B}).
\]
The affine functions defining $\phi$ are tangents to $t^2$ at $\ell/h$.
Completing the square gives
\[
  0\le t^2-\phi(t)
   =\min_{\ell=0,\ldots,h}(t-\ell/h)^2\le\frac1{4h^2},
  \qquad 0\le\phi(t)\le t^2\le t.
\]
Also $\phi(0)=0$ and $\phi(1)=1$.
At every stage, maximize $\widehat\Phi$ over $P_i$.
This is an exact rational LP: introduce $u_{j,B}\in[0,1]$, 
impose $u_{j,B}\ge2(\ell/h)y_{j,B}-(\ell/h)^2$ for $\ell=0,\ldots,h$,
and maximize $F-\beta\sum_{j,B}u_{j,B}$.
At an optimum, $u_{j,B}=\phi(y_{j,B})$.

Start with stage $0$, and discard the candidate if any stage is infeasible.
After solving stage $i$, add the copy and projected constraints defined by the sets $H_B^i$, and solve stage $i+1$. 
Stop when
\[
  \|y^{i+1}-y^i\|_2^2
  :=\sum_{j,B}(y^{i+1}_{j,B}-y^i_{j,B})^2
  \le\gamma F_{i+1}.
\]
The left side is the \emph{squared Euclidean distance} between the two collections of marginals.

\begin{lemma}\label{lem:stability}
Each candidate either becomes infeasible or satisfies the stopping condition
within $\lfloor4n/(\beta\gamma)\rfloor+1$ transitions.
In the structural case $\OPT\ge H$, selecting the feasible candidate with
largest terminal $F$ gives $F\ge(1-3\eta)\OPT$.
\end{lemma}

\begin{proof}
Define $\Phi(y):=F-\beta\|y\|_2^2$, where $\|y\|_2^2=\sum_{j,B}y_{j,B}^2$, and put $\xi:=\beta d_y/(4h^2)$.
The approximation above gives
\[
  0\le\widehat\Phi(y)-\Phi(y)\le\xi,
  \qquad
  \xi\le\frac{\beta\gamma^2}{64d_y}\le\frac{\beta\gamma}{64}.
\]
Let $V_i:=\widehat\Phi(y^i)$ be the optimal value at stage $i$.
Both $y^i$ and $y^{i+1}$ belong to $P_i$, so their midpoint does as well.
The identity
\[
  \Phi\bigl((y^i+y^{i+1})/2\bigr)
  =\frac{\Phi(y^i)+\Phi(y^{i+1})}2
   +\frac\beta4\|y^{i+1}-y^i\|_2^2
\]
and optimality of $y^i$ imply
\[
  V_i\ge\widehat\Phi\bigl((y^i+y^{i+1})/2\bigr)
      \ge\frac{V_i+V_{i+1}}2-\xi
         +\frac\beta4\|y^{i+1}-y^i\|_2^2.
\]
Consequently,
\[
  V_i-V_{i+1}\ge\frac\beta2\|y^{i+1}-y^i\|_2^2-2\xi.
\]
The density constraint implies $F_{i+1}\ge1$.
Every failed stopping test therefore decreases the objective by at least $\beta\gamma/4$. 
Since $0\le\widehat\Phi(y)\le F\le n$, there can be at most $\lfloor4n/(\beta\gamma)\rfloor$ such failures. 
This proves the transition bound, regardless of how ties between optimal solutions are broken.

For the candidate certified by Lemma~\ref{lem:partition}, the reference schedule remains feasible at every stage. 
Its marginals are zero or one, so its objective value is $(1-\beta)|O'|$.
The terminal solution thus satisfies $F\ge\widehat\Phi(y)\ge(1-\eta)|O'|$.
Selecting the largest terminal $F$ gives $F\ge(1-\eta)(1-2\eta)\OPT\ge(1-3\eta)\OPT$.
\end{proof}

At a successful transition $i\to i+1$, use $x:=x^{[i+1]}$ and its projected distributions, 
with the tracked sets $H_B:=H_B^i$ and marginals $y:=y^{i+1}$. 
In particular, the tracked sets are those used to build this copy, 
not sets recomputed from the new solution.
Let $F:=F_{i+1}$ and consider the current set $J_{\rm miss}$ from Section~\ref{sec:rounding}.
For each $j\in J_{\rm miss}$, choose a boundary $B(j)$ where $j$ was untracked.
Then $y^i_{j,B(j)}<1/8$, whereas goodness in the current solution gives $y^{i+1}_{j,B(j)}\ge1/4$.
Each such job therefore contributes more than $1/64$ to a distinct coordinate of the squared distance. Hence
\[
  \frac{|J_{\rm miss}|}{64}
  \le\sum_{j\in J_{\rm miss}}
       (y^{i+1}_{j,B(j)}-y^i_{j,B(j)})^2
  \le\gamma F,
  \qquad |J_{\rm miss}|\le64\gamma F\le\eta F.
\]
This verifies the hypothesis of Theorem~\ref{thm:rounding}.
Both outputs of Algorithm~\ref{alg:rounding} use this selected configuration copy; 
the coordinated output first samples projected states and 
then samples full configurations conditionally, as in Lemma~\ref{lem:gluing}.

\subsection{Exact pricing with required jobs}\label{sec:pricing}

We now show how to solve each stage without listing its configurations.
The \emph{dual LP} has variables corresponding to constraints of our relaxation 
and one inequality for each configuration variable.
Checking these inequalities reduces to pricing: 
finding a configuration that maximizes the expression determined by the dual coefficients.
For a block $B=[s,t)$, a tracked set $H_B$ with $|H_B|\le8K$, 
rational coefficients $w_j$, and a rational table $\alpha(A)$ for $A\subseteq H_B$,
the required problem is
\begin{equation}\label{eq:pricing}
  \max_{C\in\Ccal(B)}
       \left\{\sum_{j\in C}w_j+\alpha(C\cap H_B)\right\}.
\end{equation}
The coefficients $w_j$ and $\alpha(A)$ are auxiliary scores from the dual;
they are not probabilities, and either may be negative.
Only jobs eligible in $B$ are available.

Enumerate $A\subseteq H_B$ with $|A|\le K$.
To impose $C\cap H_B=A$, require every job of $A$ and exclude every job of $H_B\setminus A$. 
The term $\alpha(A)$ is now constant, leaving an additive objective with prescribed job identities.
This is the extension needed beyond the additive pricing algorithm in
\citep[Lemmas 21 and 30]{ArmbrusterGrandoniTinguelyWiese2026}.

Set $\kappa:=\min\{K,n\}$. Use a \emph{perfect hash family}: 
a family of maps from the jobs to colors $\{1,\ldots,\kappa\}$ such that 
every set of at most $\kappa$ jobs has distinct colors under at least one map.
Such a family of size $2^{O(\kappa)}\log^{O(1)}(n+1)$ can be constructed 
in $2^{O(\kappa)}\poly(n)$ time \citep[Section 4]{AlonYusterZwick1995}.
A set with distinct colors is \emph{colorful}.
Using $\kappa$ colors does not alter the structural bound $K$ or its density constraints.
For a fixed coloring, discard it if two required jobs have the same color.
Otherwise, reserve the color of each required job for that identity by excluding all other jobs of that color. 
Let $c(j)$ denote the color of $j$, 
and let $M:=\{c(j):j\in A\}$ be the set of required colors.

First consider one machine. For $U\subseteq\{1,\ldots,\kappa\}$ and integer $v\in[s,t]$, 
let $\operatorname{DP}(U,v)$ be the maximum weight of a schedule inside $[s,v)$ that uses only colors in $U$, 
uses each at most once, and includes every required job whose color lies in $U$.
Colors of optional jobs may remain unused.
Define the latest start of $j$ allowed by $v$ and its deadline as
$u_j(v):=\min\{v,d_j\}-p_j$, and put
\[
  b(U):=\begin{cases}
    0,&U\cap M=\varnothing,\\
    -\infty,&U\cap M\ne\varnothing.
  \end{cases}
\]
Then
\[
  \operatorname{DP}(U,v)
  =\max\left\{b(U),\
    \max_{\substack{j\text{ allowed},\ c(j)\in U\\
                     u_j(v)\ge\max\{s,r_j\}}}
    \left[w_j+\operatorname{DP}\bigl(U\setminus\{c(j)\},u_j(v)\bigr)\right]
    \right\}.
\]
The initial values are $\operatorname{DP}(\varnothing,v)=0$ and $\operatorname{DP}(U,s)=b(U)$; 
an infeasible entry has value $-\infty$.
Evaluate the entries in increasing order of $|U|$.

To verify the recurrence, consider the last job $j$ of a nonempty feasible schedule. 
Its preceding jobs finish by $u_j(v)$ and use no job of color $c(j)$.
Conversely, a predecessor schedule counted by the displayed entry can be followed by $j$ starting at $u_j(v)$, 
which respects its release time, deadline, and block boundaries. Induction on $|U|$ proves correctness.
The value $b(U)=-\infty$ prevents the empty schedule from omitting a required job, 
even when its weight is negative. 
Reserving its color prevents a different job from being selected in its place.

For $m$ machines, assign each of the $\kappa$ colors to exactly one machine.
Run the one-machine DP on each machine's color set at time $t$, sum the values,
and maximize over all $m^\kappa$ assignments.
The color sets are disjoint, so at most $\kappa\le K$ jobs are selected in total. 
Every colorful feasible configuration is represented by one such assignment. 
Moreover, every feasible configuration with $C\cap H_B=A$ is colorful under some map in the family, 
and the reservations exclude none of its jobs under that map. 
Maximizing over the assignments, colorings, and subsets $A$, 
and adding $\alpha(A)$, therefore solves~\eqref{eq:pricing} exactly.

The one-machine DP has at most $2^\kappa(T+1)$ entries and $O(n)$ transitions per entry. 
Recording the maximizing choices also produces a feasible witness schedule. 
Including all subsets $A$ and machine assignments, 
the running time is $(2m)^{O(K)}\poly(n,T,N,\Lambda)$, 
where $\Lambda$ is the encoding length of the supplied coefficients and table.

To apply this routine to stage $i$, consider a variable $x^{[k]}_{B,C}$ with $k\ge1$. 
It occurs only in its normalization equation, the equations matching the common placement marginals, 
and the projected-state equations for copy $k$.
Thus its configuration-dependent dual expression has the form
\[
  \sum_{j\in C}w^{[k]}_{j,B}
  +\alpha^{[k]}_B(C\cap H_B^{k-1}),
\]
with the normalization contributing a constant independent of $C$.
The base copy $x^{[0]}$ has only the additive term.
Consequently, the routine above checks all configuration inequalities of the dual, 
despite the many retained groups of projected constraints.
The inequalities associated with the other variables can be checked directly,
since these variables are explicitly listed.
Together, these checks give an exact \emph{separation oracle}: 
for any proposed dual solution, they either verify feasibility or return a violated inequality.

There are $O(n/\eta^2)$ additions by Lemma~\ref{lem:stability}.
Each introduces $O(|\B|^3 2^{24K})$ projected tuples and their consistency equations; 
the objective introduces $d_y(h+1)$ tangent inequalities.
The total explicitly represented system has size $2^{O(K)}\poly(n,|\B|,1/\eta,N)$.
Its coefficients have polynomial encoding length, 
and normalization together with the bounds on $y$ and $u$ bounds every variable of the stage LP.
The dual has polynomially many variables, 
and the separation oracle returns an actual dual inequality with polynomial encoding length.
Exact rational optimization and recovery of a basic solution of the stage LP
\citep{GroetschelLovaszSchrijver1988}
therefore detect infeasibility or recover an optimal solution 
whose number of positive entries and rational encoding length are polynomial in the explicit system size.
The recovered configurations retain the witness schedules supplied by pricing.
The total time over all stages is $(2m)^{O(K)}\poly(n,T,N)$ for fixed $m,\eta$.

The same DP handles the exact small-optimum tests: use the interval $[0,T)$,
make all input jobs available, use the target cardinality as the cap, 
and take unit weights, an empty tracked set, and $\alpha(\varnothing)=0$.
A target is feasible exactly when the capped optimum attains it.
The increasing target tests and their stopping rule are specified in Appendix~\ref{app:partition}.
Algorithm~\ref{alg:whole} assembles both implementations; 
the complementary rounding and the choice of $\eta$ are justified next.

\begin{algorithm}[t]
\caption{Throughput approximation for fixed $m$}
\label{alg:whole}
\begin{algorithmic}[1]
\Statex \textbf{Input:} An instance, accuracy $\varepsilon>0$, and the polynomial
or pseudopolynomial running-time requirement.
\Statex \textbf{Output:} A feasible schedule.
\State Delete individually infeasible jobs; if none remain, return the empty
schedule. Set $\eta=1/\lceil400/\min\{1,\varepsilon\}\rceil$.
\State Set $\Delta=1$ in the polynomial case and
$\Delta=4\eta^{-4}(\lceil\log_2(T+1)\rceil+1)$ otherwise.
Compute $K_0,\Gamma,H$ from Lemma~\ref{lem:partition}.
\State Test target cardinalities $1,\ldots,H$ using the exact routine for the
chosen running time, with all input jobs on $[0,T)$.
At the first infeasible target, return the largest schedule found.
If every target is feasible, retain the schedule of $H$ jobs.
\State Construct the candidate partitions from Lemma~\ref{lem:partition}.
For each candidate, solve the full consistency LP maximizing $F$ in the polynomial case, 
or run the projected procedure of Lemma~\ref{lem:stability}
in the pseudopolynomial case. Discard infeasible candidates.
\State Select the candidate with largest terminal $F$.
In the projected case, retain its newest configuration copy and the projected distributions attached to that copy.
Apply Algorithm~\ref{alg:rounding} to the selected solution.
\State In the pseudopolynomial case, also apply the complementary global-job rounding of Appendix~\ref{app:complementary} 
using the same configuration distributions and stored schedules.
\State \Return the largest schedule computed.
\end{algorithmic}
\end{algorithm}

\subsection{Completion of the guarantees}\label{sec:completion}

For the selected solution, Theorem~\ref{thm:rounding} gives
\[
  \E\left[\max\{A_\times,A_{\rm keep}\}\right]
  \ge L+\frac{13}{17}G-8\eta F
  =F-\frac4{17}G-8\eta F.
\]
Since $0\le G\le F$, this is at least $(13/17-8\eta)F$.
In the pseudopolynomial case, also use the complementary global-job rounding from \citep{ArmbrusterGrandoniTinguelyWiese2026}, 
as recorded in Appendix~\ref{app:complementary}.
Applied to the same selected configuration distributions and their stored schedules, 
it returns a feasible schedule of cardinality $A_{\rm glob}$ with $\E[A_{\rm glob}]\ge G-5\eta F$.
Returning the largest of the three schedules therefore gives
\begin{align*}
  \E\left[\max\{A_\times,A_{\rm keep},A_{\rm glob}\}\right]
  &\ge\max\left\{F-\frac4{17}G,\ G\right\}-8\eta F  
  \ge\left(\frac{17}{21}-8\eta\right)F.
\end{align*}
For the last inequality, put $g:=G/F\in[0,1]$.
The larger of $1-4g/17$ and $g$ is minimized at their intersection
$g=17/21$, where both equal $17/21$.
No independence between the output schedules is needed.

To bound the running time for every random outcome, implement each draw using a fixed number of unbiased bits.
Let $N_{\rm draw}$ bound the total number of draws and 
$N_{\rm supp}$ bound the number of positive-probability outcomes in any distribution used by a draw.
Here $N_{\rm draw}=O(|\B|+n)$: block sampling, sampling along the triangulation,
and conditional reconstruction use $O(|\B|)$ draws, 
while complementary rounding uses at most $n$ job assignments.
The recovered configuration distributions and the explicitly represented joint distributions provide the required finite supports.

For each draw, including a conditional draw, round its probabilities down to multiples of $2^{-b_{\rm rand}}$ 
and assign the remaining probability to one outcome that originally had positive probability, where
\[
  b_{\rm rand}:=
  \left\lceil\log_2\left(
     \frac{4N_{\rm supp}N_{\rm draw}n}{\eta}\right)\right\rceil.
\]
These probabilities can be sampled with exactly $b_{\rm rand}$ bits.
Their \emph{total variation distance}, half the sum of the absolute probability differences, is at most
$N_{\rm supp}2^{-b_{\rm rand}}\le\eta/(nN_{\rm draw})$.
No new outcome receives positive probability, 
so subsequent conditioning remains defined and every output schedule remains feasible.
As long as earlier outcomes agree, the exact and approximate next draws can be realized together 
so that they disagree with probability at most their total variation distance. 
A union bound over the draws then bounds the probability that 
the entire computation differs from exact sampling by $\eta/n$.
Since the returned cardinality lies in $[0,n]$, its expectation decreases by at most $\eta$.

\begin{proof}[Proof of Theorem~\ref{thm:main}]
For fixed $m$, the selected solution satisfies $F\ge(1-3\eta)\OPT$ in both implementations. 
Let $c=13/17$ in the polynomial case and $c=17/21$ in the pseudopolynomial case.
In the structural case, $\OPT\ge1$, and the expected throughput after the sampling approximation is at least
\[
  (c-8\eta)F-\eta
  \ge\bigl[(c-8\eta)(1-3\eta)-\eta\bigr]\OPT
  \ge(c-14\eta)\OPT.
\]
Here $c-8\eta>0$ and $c\le1$.
The small-optimum case is solved exactly.
With $\eta=1/\lceil400/\min\{1,\varepsilon\}\rceil$, both choices of $c$ satisfy
\[
  \frac1{c-14\eta}-\frac1c
  =\frac{14\eta}{c(c-14\eta)}<32\eta<\varepsilon.
\]
This gives the approximation ratios $17/13+\varepsilon$ and $21/17+\varepsilon$.
For fixed $m,\varepsilon$, the respective running times, 
including the exact target tests, are $n^{O_{m,\varepsilon}(1)}\poly(N)$ and $(nT)^{O_{m,\varepsilon}(1)}\poly(N)$.

Finally, remove the assumption that $m$ is fixed.
We may replace $m$ by $\min\{m,n\}$ without changing the optimum.
For every fixed $\delta>0$, the algorithm of \citep[Theorem 1.2]{ImLiMoseley2020} runs in polynomial time 
with arbitrary binary time data and retains expected fraction
$1-O(\sqrt{\log m/m})-\delta$ of the optimum for unweighted jobs.
Set $\delta:=\varepsilon/[2(1+\varepsilon)]$ and choose a constant $m_\varepsilon$ such that 
the $O(\sqrt{\log m/m})$ term is at most $\delta$ for all $m\ge m_\varepsilon$.
For these machine counts, the retained fraction is at least $1-2\delta=1/(1+\varepsilon)$, 
which suffices for both claimed factors.
For $m<m_\varepsilon$, use Algorithm~\ref{alg:whole}.
Its exponents now depend only on $\varepsilon$, 
giving $n^{O_\varepsilon(1)}\poly(N)$ and $(nT)^{O_\varepsilon(1)}\poly(N)$, respectively.
\end{proof}

\section{Conclusion}\label{sec:conclusion}

For unweighted non-preemptive throughput maximization on identical machines,
we improve the approximation ratios of \citep{ArmbrusterGrandoniTinguelyWiese2026}
from $4/3+\varepsilon$ to $17/13+\varepsilon$ in polynomial time and 
from $5/4+\varepsilon$ to $21/17+\varepsilon$ in pseudopolynomial time, 
for every fixed $\varepsilon>0$. 
Both algorithms are randomized and always return feasible schedules; 
their approximation guarantees hold in expectation. 
The improvement combines a stronger analysis of maximum job--slot matching 
through crossing exchanges with coordinated sampling for a noncrossing family of boundary edges.
The matching and coordinated rounding procedures use the same selected configuration solution, 
allowing their guarantees to be combined.
The pseudopolynomial bound additionally uses the complementary global-job rounding of \citep{ArmbrusterGrandoniTinguelyWiese2026}.
Consistency on triples of blocks enables the coordination, 
while projected consistency and configuration copies support the pseudopolynomial implementation 
when the configuration cap grows logarithmically with the time horizon.

Further improvements may come from strengthening either of the two rounding guarantees.
A promising direction is to analyze exchanges among larger sets of jobs with crossing boundary edges, 
allowing the choice of exchanges to depend on the sampled configurations.
The aim is to prove a larger net gain for the existing maximum matching,
accounting for conflicts and the fractional interior assignments that must be erased.
Another direction is to coordinate crossing boundary families whose
graph on blocks admits a tree decomposition with bags of bounded size.
Consistent joint distributions on the bags would extend the present conditional-sampling argument.
The main tasks are to select families that reduce the sum of $\min\{a_j,b_j\}$ over uncoordinated global jobs 
and to optimize the required distributions within the stated running-time bounds.

\section*{Acknowledgements}

Both authors were supported by the Office of Research at Georgia Southern University.

\section*{Statements and Declarations}


\noindent\textbf{Competing interests.}
The authors have no relevant financial or non-financial interests to disclose.

\vspace{3pt}
\noindent\textbf{Author contributions.}
Weitian Tong contributed to conceptualization, formal analysis, methodology, project administration, 
supervision, validation, writing--original draft, and writing--review and editing. 
Yao Xu contributed to formal analysis, methodology, validation, and writing--review and editing.

\vspace{3pt}
\noindent\textbf{Data availability.}
No datasets were generated or analyzed in this study.  
All information needed to support the theoretical results is contained in the manuscript and the cited references.

\vspace{3pt}
\noindent\textbf{Declaration of AI-assisted technologies.}
We used OpenAI ChatGPT to assist with preparing computational checks, drafting and revising the manuscript.  
The authors independently verified the correctness and originality of all mathematical claims, proofs, computations, and references, 
and take full responsibility for the final contents of the paper.


\appendix
\section{The dense partition specialization}\label{app:partition}

We prove Lemma~\ref{lem:partition} by modifying the merging step of
\citep[Lemmas 6 and 28]{ArmbrusterGrandoniTinguelyWiese2026}. 
The initial partition is unchanged; 
the larger value of $\Gamma$ gives the superblock density required by Lemma~\ref{lem:interior}.

Fix $m,\eta,\Delta$ as in Lemma~\ref{lem:partition}, and first assume $\OPT\ge H$. 
The construction of \citep[Lemmas 4 and 27]{ArmbrusterGrandoniTinguelyWiese2026} gives 
a partition $\B_0$ into \emph{elementary blocks} and an aligned reference schedule $O''$
such that
\[
  |O''|\ge(1-\eta)\OPT,
  \qquad |O''|\ge\frac{m|\B_0|}{\eta},
\]
with at most $K_0$ executions per elementary block across all machines. 
Choose the integer $K_0$, depending only on $m,\eta$, 
large enough that $H\ge(6m/\eta)^3$ already for $\Delta=1$; 
this is the threshold required by the imported construction. 
For fixed $m,\eta$, this construction runs in polynomial time with binary time data. 
Its greedy queries find the earliest feasible execution within an interval $[a,b)$: 
at current time $u$, job $j$ finishes earliest at $\max\{a,r_j,u\}+p_j$, 
which must be at most $\min\{b,d_j\}$. 
The fixed-machine reduction places the machines' time horizons consecutively, 
giving job $j$ the windows $[r_j+(i-1)T,d_j+(i-1)T)$ for $i=1,\ldots,m$.
Checking these windows uses binary arithmetic and does not enumerate start times.

Recall that $\Gamma=2K_0/\eta^7$ and $\eta=1/q$, so $\Gamma$ is an integer. 
Order the elementary blocks from left to right and append formal entries representing $[T,T)$ until
the list length is a multiple of $2\Delta\Gamma^q$. 
These empty entries are used only for grouping; 
they are not blocks in the resulting partitions. Let $b$ be the length of the padded list. 
The elementary bound on the number of executions gives
\[
  |\B_0|\ge\frac{|O''|}{K_0}
  \ge5(1-\eta)\Delta\Gamma^q\ge4\Delta\Gamma^q.
\]
Padding adds fewer than $2\Delta\Gamma^q$ entries, so $b\le3|\B_0|/2$.

For each $\ell=1,\ldots,q$, group consecutive entries in sets of $2\Delta\Gamma^{\ell-1}$,
and combine $\Gamma$ consecutive groups at a time. 
The nonempty unions at these two scales form the block partition $\B_\ell$ and the superblock partition $\Scal_\ell$, respectively. 
Their execution bounds and sizes satisfy
\[
  K_\ell:=2\Delta K_0\Gamma^{\ell-1}\le\Delta\Gamma^q,
  \qquad
  |\B_\ell|\le\frac{b}{2\Delta\Gamma^{\ell-1}},
  \qquad
  K_\ell|\Scal_\ell|\le\frac{K_0b}{\Gamma}
                       =\frac{\eta^7b}{2}.
\]
In particular, $\Scal_\ell=\B_{\ell+1}$ for $\ell<q$. 
Thus each superblock at one level is a block at the next level.

Consider an execution of $O''$ that violates eligibility at level $\ell$. 
Its containing block is neither its release block nor its deadline block, 
so that entire block lies inside the job's window. 
Its containing superblock is not spanned by the job and therefore contains a window endpoint. 
At any finer level, the superblock containing the execution lies within its level-$\ell$ block and is consequently spanned. 
At any coarser level, the block containing the execution contains its level-$\ell$ superblock and the same window endpoint, 
making it the job's release or deadline block. Hence an execution violates eligibility at at most one level.

Among the $q$ levels, some level $\ell^*$ therefore has at most $|O''|/q=\eta|O''|$ ineligible reference executions. 
Delete these executions and call the remaining schedule $O'$. 
Then $O'$ is aligned and eligible for $(\B_{\ell^*},\Scal_{\ell^*})$, has at most $K_{\ell^*}$ executions per block, 
and satisfies $|O'|\ge(1-\eta)|O''|\ge(1-2\eta)\OPT$.
The construction outputs all $q$ partitions; the choice of $\ell^*$ and the reference deletions are used only in the proof.

The elementary density and the padding bound give
\[
  |O'|\ge(1-\eta)|O''|
  \ge\frac{2m(1-\eta)b}{3\eta}
  \ge\frac{mb}{2\eta},
\]
where the last inequality uses $\eta\le1/4$. 
Combining this with the counting bounds above yields
\[
  \frac{m\Delta|\B_{\ell^*}|}{\eta}
  \le\frac{mb}{2\eta\Gamma^{\ell^*-1}}\le|O'|,
  \qquad
  \frac{K_{\ell^*}|\Scal_{\ell^*}|}{\eta^8}
  \le\frac{b}{2\eta}\le|O'|.
\]
These are the two required density inequalities. 
Since $q$ and $\Gamma$ depend only on $m,\eta$, 
constructing the padded list and all candidate partitions takes time polynomial in $N$ and $\Delta$.

For the exact case, use the target tests in Algorithm~\ref{alg:whole}, 
whose routines are described in Section~\ref{sec:implementation}. 
For a target cardinality $h_{\rm tar}\in\{1,\ldots,H\}$, 
a feasible schedule of at least $h_{\rm tar}$ jobs contains a feasible $h_{\rm tar}$-job subset. 
Thus, at the first infeasible target $h_{\rm tar}$, the previously computed schedule of $h_{\rm tar}-1$ jobs is optimal. 
A target above $n$ is immediately infeasible, while success through $H$ certifies $\OPT\ge H$. 
The tests use all input jobs over $[0,T)$, with no partition restrictions. 
There are at most $H$ tests, with caps at most $H$, 
giving the total bounds $n^{O(H)}\poly(N)$ and $(2m)^{O(H)}\poly(n,T,N)$ stated in Lemma~\ref{lem:partition}.

\section{Complementary global-job rounding}\label{app:complementary}

We verify the complementary rounding guarantee used in Section~\ref{sec:completion},
following \citep[Lemmas 22--25 and 34]{ArmbrusterGrandoniTinguelyWiese2026}. 
Fix the selected configuration distributions $x_{B,C}$ and their stored schedules, 
using the logarithmic choice of $\Delta$ from Section~\ref{sec:implementation}. 
The argument requires only feasibility and the density bound $F\ge m\Delta|\B|/\eta$. 
In the projected implementation, all quantities below come from the selected configuration copy and its stored schedules.

For every machine $i\in\{1,\ldots,m\}$, define
\[
  y_{j,B,i}:=\sum_{C\in\Ccal(B)}x_{B,C}
    \mathbf1\{\text{the stored schedule for }C
                      \text{ places }j\text{ on machine }i\}.
\]
Then $\sum_i y_{j,B,i}=y_{j,B}$ and $\sum_{B,i}y_{j,B,i}=y_j\le1$. 
Independently for each global job $j$, assign it to at most one pair $(B,i)$, 
choosing that pair with probability $(1-2\eta)y_{j,B,i}$ and leaving it unassigned with the remaining probability. 
Local jobs are not used.

For an interval $I=[a,b)$, write $|I|:=b-a$, and set $I_{j,B}:=B\cap[r_j,d_j)$. 
An assignment of $j$ to $B$ is \emph{long} if $p_j>\eta^4|I_{j,B}|$, and \emph{short} otherwise. 
This classification depends on the assigned block. Discard all long assignments. 
To bound their expected number, fix $B=[s,t)$ and a machine $i$. 
For an integer $h\ge0$, consider long global jobs with $y_{j,B,i}>0$ and $2^h\le p_j<2^{h+1}$.
Their intersections $I_{j,B}$ have length less than $2^{h+1}/\eta^4$. 
Since the jobs are global, each such intersection has left endpoint $s$ or right endpoint $t$. 
Their executions must therefore lie in two end portions of $B$ 
whose total length is at most $2^{h+2}/\eta^4$. 
In any stored schedule, the executions of these jobs on machine $i$ are disjoint and each has length at least $2^h$, 
so there are at most $4/\eta^4$. 
Averaging over configurations and summing over $h=0,\ldots,\lfloor\log_2T\rfloor$ gives
\[
  \sum_{\substack{j\text{ global}\\p_j>\eta^4|I_{j,B}|}}y_{j,B,i}
  \le\frac{4(\lfloor\log_2T\rfloor+1)}{\eta^4}\le\Delta.
\]
The expected total number of discarded long assignments is consequently at most $m\Delta|\B|$.

Fix a pair $(B,i)$ with $B=[s,t)$, and let $J_{B,i}$ be its short assigned jobs. 
The \emph{load} of a set of jobs is their total processing time. 
Retain $j\in J_{B,i}$ exactly when the following applicable load tests hold:
\[
\begin{aligned}
  \sum_{k\in J_{B,i}}p_k&\le t-s,\\
  \sum_{\substack{k\in J_{B,i}\\d_k\le d_j}}p_k&\le d_j-s
       &&\text{if }s<d_j<t,\\
  \sum_{\substack{k\in J_{B,i}\\r_k\ge r_j}}p_k&\le t-r_j
       &&\text{if }s<r_j<t.
\end{aligned}
\]
All tests use $J_{B,i}$ before any deletions caused by these tests. 
Each test compares the processing time of jobs that must use an interval with its length. 
A global job has at most one window endpoint strictly inside $B$, 
so each job faces at most two tests.

Condition on a short job $j$ being assigned to $(B,i)$ with positive probability. 
For either test applied to $j$, let $I$ be the tested interval and $\ell:=|I|$. 
Every job that can contribute to the test has its usable interval $I_{k,B}$ contained in $I$. 
For example, a job in the deadline test has $d_k\le d_j<t$ and, because it is global, $r_k\le s$. 
Thus $I_{k,B}=[s,d_k)\subseteq[s,d_j)$; the release test is symmetric. 
Shortness therefore gives $p_k\le\eta^4\ell$ for every counted job. 
Averaging the stored schedules shows that the sum of $p_k y_{k,B,i}$ over these jobs is at most $\ell$, 
since all their executions must lie in $I$. 
All other jobs' assignments remain independent under the conditioning, so
\[
  \E[\text{tested load}\mid j\text{ assigned to }(B,i)]
  \le\eta^4\ell+(1-2\eta)\ell\le(1-\eta)\ell.
\]

Under this conditioning, let $X$ be the tested load divided by $\eta^4\ell$, 
and put $\mu_0:=(1-\eta)/\eta^4$. 
Then $X$ is a sum of independent variables in $[0,1]$, 
including $j$'s fixed contribution, and $\E[X]\le\mu_0$. For $\theta\ge0$, 
independence and convexity of the exponential on $[0,1]$ give
\[
  \E[e^{\theta X}]
  \le\exp\bigl((e^\theta-1)\E[X]\bigr)
  \le\exp\bigl((e^\theta-1)\mu_0\bigr).
\]
The \emph{Chernoff bound} for sums of independent bounded variables, 
applied at threshold $1/\eta^4=(1+\eta/(1-\eta))\mu_0$, therefore gives
\[
  \Pr(\text{test fails}\mid j\text{ assigned to }(B,i))
  \le\exp\!\left(-\frac{1}{3(1-\eta)\eta^2}\right)\le\eta.
\]
Here $\eta/(1-\eta)\le1$ and $\eta\le1/100$. 
A union bound over the at most two tests
shows that a short assigned job survives with conditional probability at least $1-2\eta$.

To schedule the surviving jobs on machine $i$ within $B=[s,t)$, 
first place those with $s<d_j<t$ in increasing deadline order, consecutively from $s$. 
Place those with $s<r_j<t$ in increasing release order, 
consecutively in a suffix ending at $t$. 
The remaining jobs span $B$; place them in any order in the gap between these two groups. 
The groups are disjoint because all assigned jobs are global, 
and the total load test guarantees enough space. 
Let $J'_{B,i}$ denote the surviving jobs. 
For a job in the first group, its completion time is at most
\[
  s+\sum_{\substack{k\in J'_{B,i}\\d_k\le d_j}}p_k\le d_j.
\]
For a job in the second group, its start time is at least
\[
  t-\sum_{\substack{k\in J'_{B,i}\\r_k\ge r_j}}p_k\ge r_j.
\]
These inequalities follow from the passed tests, since deleting jobs only decreases loads,
and remain valid when deadlines or release times are tied. 
The other window endpoint imposes no additional restriction within $B$. 
Thus every block schedule is feasible. 
Since each job was assigned to at most one pair $(B,i)$, combining these schedules is feasible as well.

The expected number of initial assignments is $(1-2\eta)G$. 
Subtracting the expected long assignments and applying conditional survival to the remaining assignments yields
\[
\begin{aligned}
  \E[A_{\rm glob}]
  &\ge(1-2\eta)\bigl((1-2\eta)G-m\Delta|\B|\bigr)\\
  &\ge(1-2\eta)^2G-m\Delta|\B|
  \ge G-4\eta G-\eta F\ge G-5\eta F.
\end{aligned}
\]
The last line uses the density bound and $G\le F$. 
All deterministic steps are polynomial in the input and the represented configuration supports; 
the procedure uses at most one random assignment per global job. 
Section~\ref{sec:completion} implements these draws with a bounded number of random bits.

\clearpage
\section{Roadmap of the paper}\label{app:roadmap}

\begingroup
\setlength{\intextsep}{5pt}
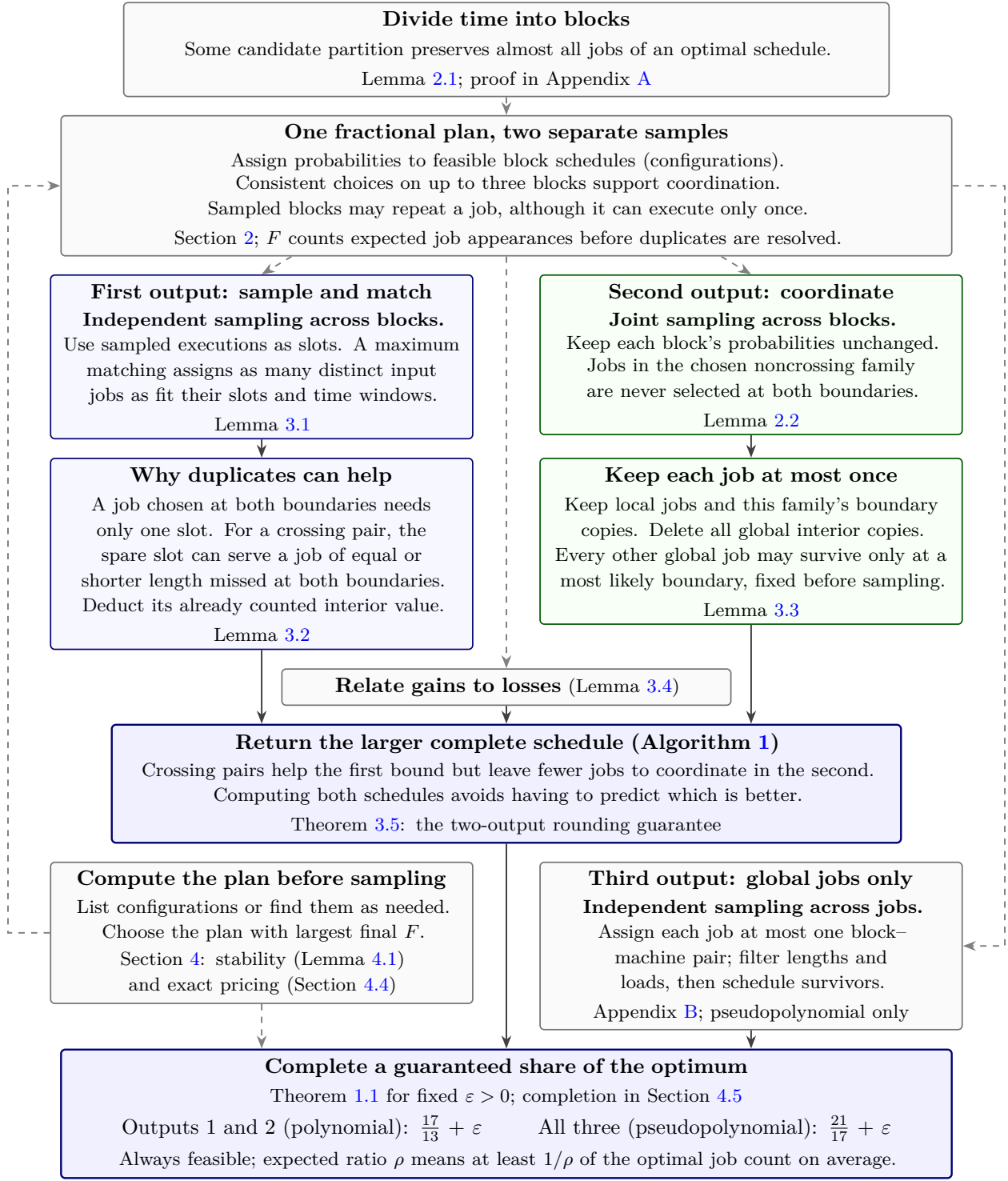
\begin{figure}[H]
\centering
\captionsetup{font=footnotesize,skip=5pt}
\resizebox{\linewidth}{!}{%
\begin{tikzpicture}[
  x=1cm,y=1cm,font=\small,
  box/.style={draw=black!55,fill=black!2,rounded corners=2pt,
    line width=.65pt,align=center,inner xsep=6pt,inner ysep=4pt},
  first/.style={box,draw=blue!45!black,fill=blue!3},
  second/.style={box,draw=green!35!black,fill=green!3},
  result/.style={box,draw=blue!45!black,fill=blue!6,line width=.9pt},
  proof/.style={-{Stealth[length=2.2mm]},draw=black!75,line width=.75pt},
  input/.style={proof,dashed,draw=black!50},
  note/.style={font=\footnotesize,align=center,fill=white,inner sep=2pt}
]

\node[box,text width=12.4cm] (partition) at (0,0) {
  \textbf{Divide time into blocks}\\[2pt]
  {\footnotesize Some candidate partition preserves almost all jobs of an optimal schedule.}\\[2pt]
  {\footnotesize Lemma~\ref{lem:partition}; proof in Appendix~\ref{app:partition}}
};
\node[box,text width=14.5cm,below=3mm of partition] (lp) {
  \textbf{One fractional plan, two separate samples}\\[2pt]
  {\footnotesize Assign probabilities to feasible block schedules (configurations).\\
  Consistent choices on up to three blocks support coordination.\\
  Sampled blocks may repeat a job, although it can execute only once.}\\[2pt]
  {\footnotesize Section~\ref{sec:prelim}; $F$ counts expected job appearances before duplicates are resolved.}
};

\node[first,text width=6.65cm,anchor=north] (interior)
  at ($(lp.south)+(-4.1,-.3)$) {
  \textbf{First output: sample and match}\\[2pt]
  {\footnotesize \textbf{Independent sampling across blocks.}\\
  Use sampled executions as slots. A maximum matching assigns as many distinct
  input jobs as fit their slots and time windows.}\\[2pt]
  {\footnotesize Lemma~\ref{lem:interior}}
};

\node[second,text width=6.65cm,anchor=north] (gluing)
  at ($(lp.south)+(4.1,-.3)$) {
  \textbf{Second output: coordinate}\\[2pt]
  {\footnotesize \textbf{Joint sampling across blocks.}\\
  Keep each block's probabilities unchanged.\\
  Jobs in the chosen noncrossing family\\
  are never selected at both boundaries.}\\[2pt]
  {\footnotesize Lemma~\ref{lem:gluing}}
};
\node[fit=(interior)(gluing),inner sep=0pt] (firstrow) {};

\node[first,text width=6.65cm,anchor=north] (exchange)
  at ($(firstrow.south)+(-4.1,-.3)$) {
  \textbf{Why duplicates can help}\\[2pt]
  {\footnotesize A job chosen at both boundaries needs only one slot.
  For a crossing pair, the spare slot can serve a job of equal or shorter length
  missed at both boundaries. Deduct its already counted interior value.}\\[2pt]
  {\footnotesize Lemma~\ref{lem:rescue}}
};
\node[second,text width=6.65cm,anchor=north] (coordination)
  at ($(firstrow.south)+(4.1,-.3)$) {
  \textbf{Keep each job at most once}\\[2pt]
  {\footnotesize Keep local jobs and this family's boundary copies.
  Delete all global interior copies. Every other global job may survive only at
  a most likely boundary, fixed before sampling.}\\[2pt]
  {\footnotesize Lemma~\ref{lem:coordination}}
};
\node[fit=(exchange)(coordination),inner sep=0pt] (secondrow) {};

\node[box,text width=7.1cm,below=3mm of secondrow] (product) {
  \textbf{Relate gains to losses} {\footnotesize (Lemma~\ref{lem:good-product})}
};
\node[result,text width=12.8cm,below=3mm of product] (rounding) {
  \textbf{Return the larger complete schedule (Algorithm~\ref{alg:rounding})}\\[2pt]
  {\footnotesize Crossing pairs help the first bound but leave fewer jobs to coordinate in the second.\\
  Computing both schedules avoids having to predict which is better.}\\[2pt]
  {\footnotesize Theorem~\ref{thm:rounding}: the two-output rounding guarantee}
};

\node[box,text width=6.65cm,anchor=north] (implementation)
  at ($(rounding.south)+(-4.1,-.3)$) {
  \textbf{Compute the plan before sampling}\\[2pt]
  {\footnotesize List configurations or find them as needed.\\
  Choose the plan with largest final $F$.}\\[2pt]
  {\footnotesize Section~\ref{sec:implementation}: stability (Lemma~\ref{lem:stability})\\
  and exact pricing (Section~\ref{sec:pricing})}
};
\node[box,text width=6.65cm,anchor=north] (global)
  at ($(rounding.south)+(4.1,-.3)$) {
  \textbf{Third output: global jobs only}\\[2pt]
  {\footnotesize \textbf{Independent sampling across jobs.}\\
  Assign each job at most one block--machine pair;
  filter lengths and loads, then schedule survivors.}\\[2pt]
  {\footnotesize Appendix~\ref{app:complementary}; pseudopolynomial only}
};
\node[fit=(implementation)(global),inner sep=0pt] (thirdrow) {};
\node[result,text width=14.5cm,below=3mm of thirdrow] (main) {
  \textbf{Complete a guaranteed share of the optimum}\\[2pt]
  {\footnotesize Theorem~\ref{thm:main} for fixed \mbox{$\varepsilon>0$}; completion in Section~\ref{sec:completion}}\\[3pt]
  Outputs 1 and 2 (polynomial): $\frac{17}{13}+\varepsilon$
  \hspace{.7cm}
  All three (pseudopolynomial): $\frac{21}{17}+\varepsilon$\\[3pt]
  {\footnotesize Always feasible; expected ratio $\rho$ means at least $1/\rho$ of the optimal job count on average.}
};

\draw[input] (partition) -- (lp);
\draw[input] ([xshift=-3.6cm]lp.south) -- (interior.north);
\draw[input] ([xshift=3.6cm]lp.south) -- (gluing.north);
\draw[input] (lp.south) -- (product.north);
\coordinate (leftrail) at (-8.35,0);
\coordinate (rightrail) at (8.35,0);
\draw[input] (implementation.west) -- (leftrail |- implementation.west)
  -- (leftrail |- lp.west) -- (lp.west);
\draw[input] (lp.east) -- (rightrail |- lp.east)
  -- (rightrail |- global.east) -- (global.east);
\draw[input] (implementation.south)
  -- ([xshift=-4.1cm]main.north);

\draw[proof] (interior) -- (exchange);
\draw[proof] (gluing) -- (coordination);
\draw[proof] (exchange.south) -- ([xshift=-4.1cm]rounding.north);
\draw[proof] (product.south) -- (rounding.north);
\draw[proof] (coordination.south) -- ([xshift=4.1cm]rounding.north);
\draw[proof] (rounding.south) -- (main.north);
\draw[proof] (global.south) -- ([xshift=4.1cm]main.north);

\end{tikzpicture}%
}
\caption{Local jobs have windows in one block; others are global.
Good global jobs have probability at least $1/4$ at each boundary.
Choose a maximal matching of crossing pairs of good jobs (boundary endpoints alternate).
The unmatched good jobs form a noncrossing family; coordinate those recorded at both boundaries.
Full consistency records all jobs.
Interior assignments and exchanges are proof devices.
Solid arrows carry bounds; dashed arrows supply inputs.}
\label{fig:paper-roadmap}
\end{figure}
\endgroup

\bibliographystyle{alpha}
\bibliography{manuscript-20261025}

\end{document}